\documentclass[aps,prx,reprint,floatfix,longbibliography]{revtex4-2}

\usepackage{graphicx}
\usepackage{amsfonts,amssymb,amsmath}
\usepackage{amsthm}
\usepackage{hyperref}

\theoremstyle{plain}
\newtheorem{theorem}{Theorem}
\newtheorem{proposition}{Proposition}
\newtheorem{lemma}{Lemma}
\newtheorem{corollary}{Corollary}
\newtheorem{definition}{Definition}
\newtheorem{remark}{Remark}

\newcommand{\Tr}{\operatorname{Tr}}

\begin{document}

\title{Task-Dependent Syndrome Memory in Quantum Sensing and State Recovery}
\author{Jianqi Sheng}
\affiliation{Department of Physics, City University of Hong Kong, Hong Kong}
\email{sheng.jq@cityu.edu.hk}

\begin{abstract}
Preserving metrological information under noise is central to quantum sensing, yet finite detectors and memories impose an unavoidable limit on how finely noise trajectories can be resolved. Using exact symmetric-logarithmic-derivative geometry, we determine when a monitored trajectory can be compressed without losing quantum Fisher information. For an explicit faithful monitored qubit family, the complete joint signal-and-noise model is recoverable from only polynomially many type records, requiring $O(\log n)$ terminal memory, whereas deferred recovery of arbitrary $n$-qubit states requires exponentially many trajectories, or $O(n)$ memory. Online correction replaces this terminal storage by an irreducible per-use readout and feedback alphabet. These results establish syndrome information as a task- and timing-dependent resource connecting quantum sensing, statistical sufficiency, and quantum error correction.
\end{abstract}
\maketitle

\section{Introduction}

Noise removes information from a quantum sensor not only by changing the probe state, but also by making distinct physical trajectories observationally indistinguishable. Quantum error correction, environment monitoring, and erasure conversion address this second mechanism by producing a readable record correlated with the noise event. Such a record can preserve metrological information even when the corresponding unflagged mixture has little or no quantum Fisher information (QFI) \cite{escher2011general,demkowicz2012elusive,zhou2018achieving,niroula2024erasure,liu2026converting}.

Resolving every microscopic event is itself a resource. An $r$-outcome monitored instrument used $n$ times generates $r^n$ ordered histories, although the statistical task may depend only on a much smaller summary. This problem is distinct from compressing the parameter-bearing quantum systems themselves \cite{yang2018population,tang2026compressing} or coarse graining their final measurement \cite{hovhannisyan2021thermometry,go2026coarse}: here the quantum outputs are retained and only a physically available classical--quantum noise record is processed. The central question is therefore which parts of a monitored trajectory must be remembered for metrological inference, complete statistical-model recovery, and universal quantum-state recovery.

We fix an orthogonal fine record produced by a monitored ancilla or accessible environment sector and allow parameter-independent classical processing of its labels. Figure~\ref{fig:resource-model} shows $n$ independent monitored uses, the label and conditional quantum output produced in each use, and the allowed compression of only the classical trajectory. Label-controlled quantum recovery is either deferred until the terminal record is available or applied online after every use; these two timings define different storage and feedback resources.

\begin{figure*}[t]
\centering
\includegraphics[width=0.8\textwidth]{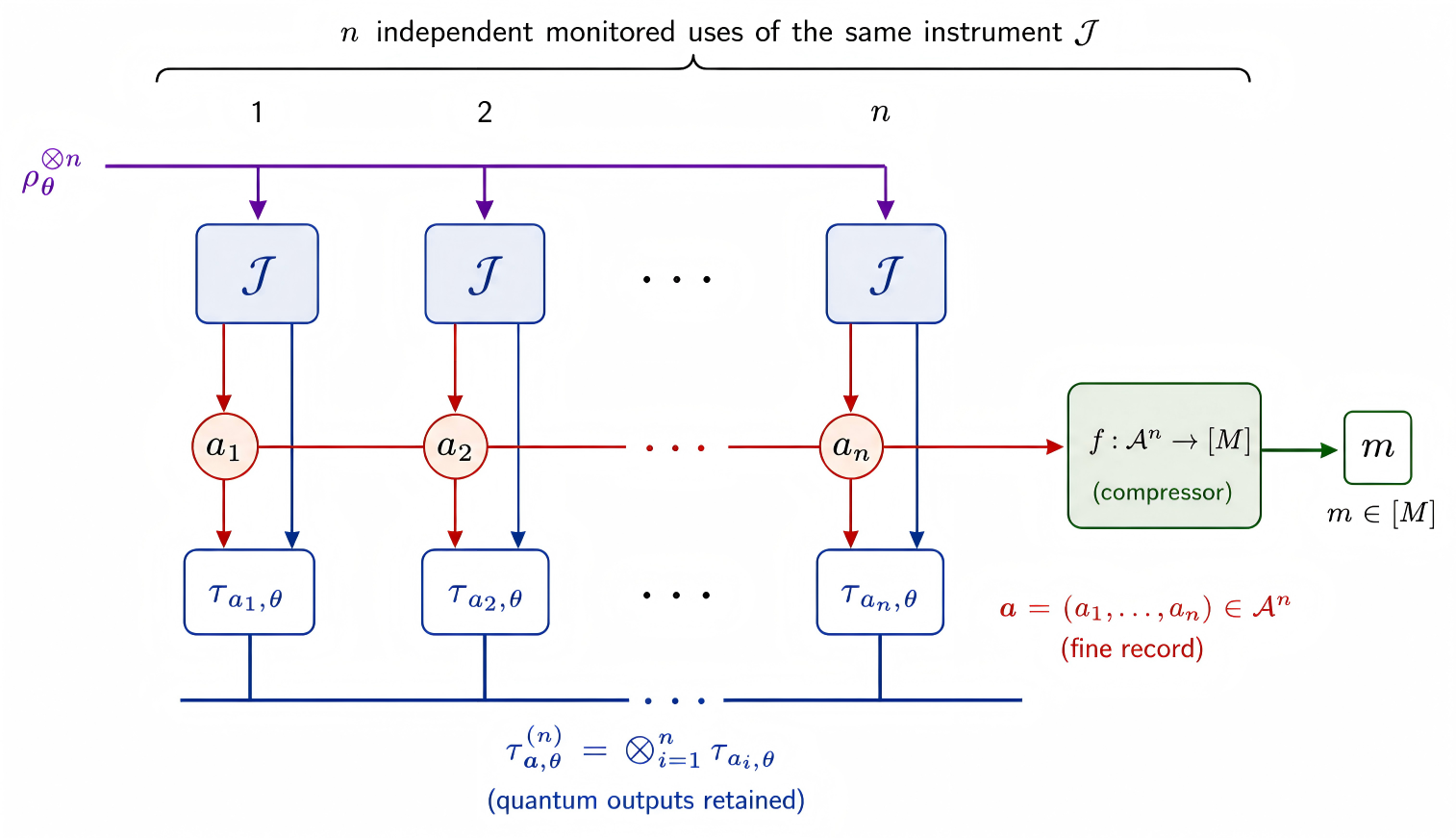}
\caption{Operational model. The same monitored instrument $\mathcal J$ acts independently on $n$ copies of $\rho_\theta$. Use $i$ produces a readable label $a_i\in\mathcal A$ and a conditional quantum output $\tau_{a_i,\theta}$. Only the ordered classical record $\boldsymbol a=(a_1,\ldots,a_n)$ is compressed to $m\in[M]$; the joint quantum output $\tau_{\boldsymbol a,\theta}^{(n)}=\bigotimes_i\tau_{a_i,\theta}$ is retained. In the explicit family below, complete joint signal-and-noise model preservation requires $O(\log n)$ terminal memory, whereas deferred recovery of arbitrary inputs requires $O(n)$.}
\label{fig:resource-model}
\end{figure*}

The analysis has three layers. First, an exact Pythagorean remainder gives the QFI lost under any finite-syndrome partition and a support-resolved zero-loss criterion. Second, for repeated faithful uses, exact permutation-invariant compression is possible precisely when the branch SLD scores share a common non-scalar part; the accumulated scalar score is then the minimal QFI-preserving statistic, and an approximate theorem controls deviations from this structure. Third, explicit monitored qubit families exhibit polynomial metrological records but exponential universal-recovery records. A graph formulation shows that online correction can recycle peak memory only by paying an irreducible syndrome-readout and feedback alphabet on every use.

Sufficient statistics and type classes are classical constructions \cite{fisher1922foundations,koopman1936sufficient,pitman1936sufficient,csiszar1998types}. The quantum content is the exact SLD criterion for forgetting trajectory order, its stability for noncommuting conditional outputs, recovery of a complete flagged quantum statistical family while the quantum outputs remain untouched, and the separation from the side information required for universal quantum recovery. A short final section embeds the syndrome identity into a general channel-recovery geometry. 

\section{Monitored records and exact information loss}
\label{sec:records-loss}

All Hilbert spaces are finite dimensional. For a positive, possibly subnormalized operator $\tau_\theta$, an SLD score $S_\theta=S_\theta^\dagger$ satisfies~\cite{braunstein1994statistical,paris2009quantum}
\begin{equation}
  \dot\tau_\theta=\frac12(S_\theta\tau_\theta+\tau_\theta S_\theta).
  \label{eq:subnormalized-sld}
\end{equation}
Its QFI contribution is $\operatorname{Tr}(\tau_\theta S_\theta^2)$. If $\tau_\theta=q_\theta\rho_\theta$, then $S_\theta=(\dot q_\theta/q_\theta)I+L_\theta$ on the support of $\rho_\theta$, where $L_\theta$ is the normalized-state SLD. Rank-deficient states are treated supportwise; faithfulness is invoked explicitly when a full operator identity is required.

Let $a\in\mathcal A=\{1,\ldots,r\}$ label the physically available fine record,
\begin{equation}
  \Omega_\theta^{\mathrm{fine}}=\bigoplus_{a=1}^r\tau_{a,\theta},
  \qquad \sum_a\operatorname{Tr}\tau_{a,\theta}=1.
  \label{eq:fine-state}
\end{equation}
A deterministic compressor $f:\mathcal A\to[M]$ produces
\begin{equation}
  \tau_{m,\theta}^{(f)}=\sum_{a:f(a)=m}\tau_{a,\theta},
  \qquad
  \Omega_\theta^{(f)}=\bigoplus_{m=1}^M\tau_{m,\theta}^{(f)}.
  \label{eq:coarse-state}
\end{equation}
Randomized readout cannot exceed the best deterministic partition by convexity of the SLD QFI; the explicit convex decomposition is given in the Supplemental Material \cite{supp}. If the fine label is stored in an orthogonal register, the coarse record is measured directly by $\Pi_m=\sum_{a:f(a)=m}|a\rangle\!\langle a|$.

\begin{proposition}[Orthogonal-pointer realization of the fine record]
\label{prop:pointer-realization}
Let $\{E_a\}_{a=1}^r$ be Kraus operators satisfying
\begin{equation}
  \sum_{a=1}^r E_a^\dagger E_a=I,
  \label{eq:pointer-completeness}
\end{equation}
and let $\{|a\rangle_A\}_{a=1}^r$ be orthonormal states of a readout ancilla.  The isometry
\begin{equation}
  V=\sum_{a=1}^r |a\rangle_A\otimes E_a
  \label{eq:pointer-isometry}
\end{equation}
followed by projective readout of $A$ produces the fine classical--quantum state
\begin{equation}
  \mathcal J(\rho_\theta)
  =\sum_{a=1}^r |a\rangle\!\langle a|_A
  \otimes E_a\rho_\theta E_a^\dagger.
  \label{eq:pointer-fine-state}
\end{equation}
For an ordered batch of $n$ independently prepared monitored uses, fresh ancillas produce
\begin{equation}
  \mathcal J(\rho_\theta)^{\otimes n}
  =\sum_{\boldsymbol a\in\mathcal A^n}
  |\boldsymbol a\rangle\!\langle\boldsymbol a|
  \otimes\bigotimes_{i=1}^n
  E_{a_i}\rho_\theta E_{a_i}^\dagger.
  \label{eq:pointer-n-use-state}
\end{equation}
Any batch projector or streaming classical update acting only on $\boldsymbol a$ therefore implements the compressors studied below without modifying the retained quantum outputs.
\end{proposition}

The proof and experimental interpretation of Proposition~\ref{prop:pointer-realization} are given in the Supplemental Material \cite{[{See Supplemental Material at }][{ for deferred proofs, structural channel results, one-shot benchmarks, the continuous-axis solution, and optimization details.}]supp} The construction makes explicit that the theory compresses a record already deposited in orthogonal pointer sectors; it does not generate such sectors from an arbitrary coherent error. Repetitive ancilla readout and quantum-error-corrected sensing provide concrete experimental realizations of the required monitored-record architecture~\cite{unden2016repetitive,xue2020repetitive}.

\begin{theorem}[Exact syndrome-compression identity]
\label{thm:exact-loss}
For every parameter-independent deterministic compressor \(f\),
\begin{equation}
  F_Q(\Omega_\theta^{\mathrm{fine}})
  -F_Q(\Omega_\theta^{(f)})
  =\sum_{a=1}^{r}
  \Tr\!\left[
    \tau_a\left(S_a-T_{f(a)}\right)^2
  \right].
  \label{eq:exact-loss}
\end{equation}
In particular, every term on the right-hand side is nonnegative.
\end{theorem}

\begin{proof}
Fix a nonempty class \(C_m=\{a:f(a)=m\}\).  Expanding the proposed residual gives
\begin{align}
 &\sum_{a\in C_m}\Tr\left[\tau_a(S_a-T_m)^2\right] \notag\\
 &=\sum_{a\in C_m}\Tr(\tau_aS_a^2)
   +\Tr(\tau_mT_m^2) \notag\\
 &\quad -\sum_{a\in C_m}
   \Tr\left[\tau_a(S_aT_m+T_mS_a)\right].
 \label{eq:loss-expand}
\end{align}
Using Eq.~\eqref{eq:subnormalized-sld} and cyclicity of the trace,
\begin{align}
 &\sum_{a\in C_m}
 \Tr\left[\tau_a(S_aT_m+T_mS_a)\right] \notag\\
 &=2\sum_{a\in C_m}\Tr(\dot\tau_aT_m)
 =2\Tr(\dot\tau_mT_m)
 =2\Tr(\tau_mT_m^2).
 \label{eq:cross-collapse}
\end{align}
Substitution into Eq.~\eqref{eq:loss-expand} yields the fine QFI contribution of \(C_m\) minus the QFI of its coarse block.  Summing over \(m\) proves Eq.~\eqref{eq:exact-loss}.  Finally, for Hermitian \(D\),
\begin{equation}
  \Tr(\tau_aD^2)
  =\Tr\!\left[(D\sqrt{\tau_a})^\dagger
  D\sqrt{\tau_a}\right]\geq0.
\end{equation}
\end{proof}

The same identity gives the exact condition for lossless compression.

\begin{corollary}[Necessary and sufficient condition for lossless compression]
\label{cor:lossless}
A partition \(f\) preserves all fine-record QFI at a fixed parameter value if and only if, for each nonempty class \(C_m\), there exists a Hermitian \(T_m\) such that
\begin{equation}
  (S_a-T_m)\sqrt{\tau_a}=0,
  \qquad a\in C_m.
  \label{eq:support-compatibility}
\end{equation}
If all \(\tau_a\) are faithful, this condition reduces to \(S_a=T_m\) within each class.  Defining \(M_{\mathrm{QFI}}(\theta)\) as the minimum number of nonempty classes in a lossless partition at the chosen parameter value, one obtains
\begin{equation}
  M_{\mathrm{QFI}}(\theta)
  =\left|\{S_a(\theta):a\in\mathcal A\}\right|
  \label{eq:distinct-sld-number}
\end{equation}
in the faithful case, after branches with identical scores are identified.
\end{corollary}

\begin{proof}
By Theorem~\ref{thm:exact-loss}, zero total loss is equivalent to the vanishing of every nonnegative summand,
\(\Tr[\tau_a(S_a-T_{f(a)})^2]=0\), which is equivalent to Eq.~\eqref{eq:support-compatibility}.  Conversely, suppose a Hermitian \(T_m\) satisfies Eq.~\eqref{eq:support-compatibility} for all \(a\in C_m\).  Since \(S_a-T_m\) is Hermitian, the same condition implies both
\((S_a-T_m)\tau_a=0\) and \(\tau_a(S_a-T_m)=0\).  Therefore
\begin{align}
  T_m\tau_m+\tau_mT_m
  &=\sum_{a\in C_m}(T_m\tau_a+\tau_aT_m)\notag\\
  &=\sum_{a\in C_m}(S_a\tau_a+\tau_aS_a)
  =2\dot\tau_m,
\end{align}
so \(T_m\) is automatically a coarse-block SLD.  The faithful-case counting statement follows because every fine SLD is then unique and two branches can share a lossless class exactly when their scores coincide.
\end{proof}

For several parameters $\boldsymbol\theta=(\theta_1,\ldots,\theta_p)$, let $S_{a,\mu}$ and $T_{m,\mu}$ be the fine and coarse SLD scores. Applying Theorem~\ref{thm:exact-loss} to every real parameter direction and polarizing yields
\begin{align}
  J^{\mathrm{fine}}_{\mu\nu}-J^{(f)}_{\mu\nu}
  &=\frac12\sum_a\operatorname{Tr}\!\left[
  \tau_a\{D_{a,\mu},D_{a,\nu}\}\right],
  \label{eq:qfim-loss}\\
  D_{a,\mu}&=S_{a,\mu}-T_{f(a),\mu}.
\end{align}
so $J^{\mathrm{fine}}-J^{(f)}\succeq0$.
\begin{corollary}[Lossless multiparameter compression]
\label{cor:multiparameter-lossless}
The full matrix is preserved exactly when
\begin{equation}
  (S_{a,\mu}-T_{f(a),\mu})\sqrt{\tau_a}=0
  \quad\text{for all }a,\mu.
  \label{eq:multiparameter-zero-loss}
\end{equation}
\end{corollary}
The normalized extended-convexity remainder~\cite{alipour2015extended}, the Pythagorean center identity, and the resulting operator-clustering variational principle are derived in the Supplemental Material \cite{supp}. They provide constructive one-shot design tools but are not required for the repeated-use rate theorem.

\section{When can the trajectory order be forgotten?}
\label{sec:order-forgetting}

Suppose one monitored use produces faithful blocks $\{\tau_{a,\boldsymbol\vartheta}\}_{a\in\mathcal A}$. For $n$ independent uses,
\begin{equation}
  \tau_{\boldsymbol a,\boldsymbol\vartheta}^{(n)}
  =\tau_{a_1,\boldsymbol\vartheta}\otimes\cdots\otimes
  \tau_{a_n,\boldsymbol\vartheta},
  \qquad \boldsymbol a\in\mathcal A^n.
  \label{eq:iid-fine-block}
\end{equation}
The first result identifies the exact structure required to discard the ordering information.

\begin{theorem}[Exact and stable order forgetting]
\label{thm:order-forgetting}
Let $n\ge2$ and let the one-use fine blocks be faithful.
\begin{enumerate}
\item If a permutation-invariant deterministic compressor preserves the full $n$-use SLD QFI matrix, then for every parameter component $\mu$ there exist a Hermitian operator $L_\mu$ and real scalars $s_{a,\mu}$ such that
\begin{equation}
  S_{a,\mu}=L_\mu+s_{a,\mu}I.
  \label{eq:necessary-common-score-form}
\end{equation}
\label{thm:common-score-necessity}
\item Conversely, under Eq.~\eqref{eq:necessary-common-score-form}, define
\begin{equation}
  \boldsymbol u(\boldsymbol a)=\sum_{i=1}^n\boldsymbol s_{a_i},
  \qquad \boldsymbol s_a=(s_{a,1},\ldots,s_{a,p}).
  \label{eq:accumulated-score}
\end{equation}
A compressor is QFI-lossless exactly when it never merges trajectories with different accumulated scores. Hence
\begin{equation}
  M_{\mathrm{QFI}}^{(n)}=|\mathcal U_n|
  \le \binom{n+r-1}{r-1}\le(n+1)^{r-1},
  \label{eq:accumulated-score-alphabet}
\end{equation}
where $\mathcal U_n$ is the set of reachable accumulated scores. If the common-score form holds throughout a parameter region, the parameter-independent type map is lossless throughout that region.
\label{thm:accumulated-score}
\item More generally, choose Hermitian $L_\mu$, set $q_a=\operatorname{Tr}\tau_a$, and define
\begin{align}
  s_{a,\mu}&=\frac{\operatorname{Tr}[\tau_a(S_{a,\mu}-L_\mu)]}{q_a},\\
  \Delta_{a,\mu}&=S_{a,\mu}-L_\mu-s_{a,\mu}I,
  \label{eq:approximate-common-score-residual}
\end{align}
with defect matrix
\begin{equation}
  \Gamma_{\mu\nu}=\frac12\sum_a\operatorname{Tr}\!\left[
  \tau_a\{\Delta_{a,\mu},\Delta_{a,\nu}\}\right].
  \label{eq:common-score-defect-matrix}
\end{equation}
Then type compression obeys
\begin{equation}
  0\preceq J_{\mathrm{fine}}^{(n)}-J_{\mathrm{type}}^{(n)}
  \preceq n\Gamma.
  \label{eq:stable-type-compression-bound}
\end{equation}
\label{thm:stable-type-compression}
\end{enumerate}
\end{theorem}

\begin{proof}
For the first statement, compare two trajectories that differ only by exchanging labels $a$ and $b$ in the first two tensor factors. Losslessness and faithfulness make their product SLDs equal, so
$D_\mu\otimes I=I\otimes D_\mu$ with
$D_\mu=S_{a,\mu}-S_{b,\mu}$. A partial trace forces
$D_\mu=\operatorname{Tr}(D_\mu)I/d$, proving Eq.~\eqref{eq:necessary-common-score-form}.

Under this form, the product SLD is
\begin{equation}
  S_{\boldsymbol a,\mu}^{(n)}
  =\sum_{i=1}^n L_\mu^{(i)}+u_\mu(\boldsymbol a)I.
  \label{eq:product-common-score}
\end{equation}
Faithfulness and Eq.~\eqref{eq:multiparameter-zero-loss} therefore imply that two trajectories can share a lossless class exactly when their accumulated scores agree. The accumulated score depends only on the type vector, giving Eq.~\eqref{eq:accumulated-score-alphabet} and a sequential update rule.

For the stability statement, the common part and accumulated scalar score form an admissible type-class center. The centered residuals on distinct tensor factors have vanishing cross terms after summation over trajectories, leaving $n\boldsymbol v^{\mathsf T}\Gamma\boldsymbol v$ in every real parameter direction $\boldsymbol v$. The complete derivation is given in the Supplemental Material \cite{supp}.
\end{proof}

The third statement is a stability result: a fixed nonzero defect gives a uniform per-use loss bound, not asymptotically vanishing total loss.

\section{Task-dependent syndrome-memory rates}
\label{sec:rates}

Polynomial QFI compression does not require identical or commuting conditional states.

\paragraph*{Noncommuting realization.}
\phantomsection\label{thm:noncommuting-qfi-rate}
For every $r\ge2$ and $n\ge1$, a faithful monitored qubit model exists whose normalized conditional outputs are pairwise noncommuting, while
\begin{equation}
  M_{\mathrm{QFI}}^{(n)}=\binom{n+r-1}{r-1},
  \qquad M_{\mathrm{state}}^{(n)}=r^n.
  \label{eq:noncommuting-rate-separation}
\end{equation}
Thus polynomial QFI compression does not rely on identical or commuting conditional quantum states.

A faithful SLD exponential family and label-dependent rotations about a common generator realize Theorem~\ref{thm:noncommuting-qfi-rate}. The conditional Bloch vectors are pairwise noncommuting, while the noise-weight scores distinguish all type vectors. The same error operators have a complete $n$-use incompatibility graph, so universal recovery still requires every trajectory. The construction and direct checks are given in the Supplemental Material \cite{supp}.

Preserving the complete ordered flagged model is stronger than preserving its QFI. The following result explains why the exact model-recovery example below has identical normalized branches in the informationally complete regime.

\paragraph*{Complete-model obstruction.}
\phantomsection\label{thm:full-model-recovery-obstruction}
Let $\Omega_{\boldsymbol\vartheta}=\bigoplus_j\tau_{j,\boldsymbol\vartheta}$ be faithful and let $X_{\boldsymbol\vartheta}=\sum_j\tau_{j,\boldsymbol\vartheta}$. Forgetting $j$ is reversible on the full family exactly when there are parameter-independent operators $A_j$ such that
\begin{align}
  \tau_{j,\boldsymbol\vartheta}&=A_jX_{\boldsymbol\vartheta}A_j^\dagger,
  \label{eq:petz-branch-factorization}\\
  \sum_jA_j^\dagger A_j&=I.
  \label{eq:petz-branch-completeness}
\end{align}
If
\begin{equation}
  \operatorname{span}_{\mathbb C}\{X_{\boldsymbol\vartheta}\}
  =\mathcal L(\mathcal H),
  \label{eq:informationally-complete-coarse-family}
\end{equation}
then necessarily
\begin{equation}
  \tau_{j,\boldsymbol\vartheta}=p_jX_{\boldsymbol\vartheta}
  \label{eq:flat-branch-necessity}
\end{equation}
with parameter-independent probabilities $p_j$. Hence an informationally complete faithful coarse family can recover forgotten labels only when the normalized branch states are identical.

The Petz characterization and a noncommuting recovery example on a proper sufficient subsystem are given in the Supplemental Material \cite{supp}. Thus nonflat full-model recovery is possible on a proper sufficient subsystem, but informational completeness forces the flat structure in Eq.~\eqref{eq:flat-branch-necessity}.

\begin{theorem}[Zero-rate model record versus positive-rate state-recovery record]
\label{thm:rate-separation}
For every integer \(r\geq2\) and every \(n\geq1\), there is a faithful \(r\)-parameter qubit model with \(r\) fine random-unitary errors per use such that
\begin{equation}
  M_{\mathrm{QFI}}^{(n)}
  =M_{\mathrm{model}}^{(n)}
  =\binom{n+r-1}{r-1},
  \qquad
  M_{\mathrm{state}}^{(n)}=r^n.
  \label{eq:exact-rate-separation}
\end{equation}
Here \(M_{\mathrm{model}}^{(n)}\) is the minimum size of a classical terminal alphabet obtainable by parameter-independent post-processing of the fine label and admitting a parameter-independent channel that reconstructs the entire fine flagged statistical family.  Likewise, \(M_{\mathrm{state}}^{(n)}\) allows arbitrary parameter-independent classical post-processing of the fine label before exact recovery of every \(n\)-qubit input state.  Hence
\begin{align}
  \log_2 M_{\mathrm{model}}^{(n)}
  &=(r-1)\log_2 n-\log_2[(r-1)!]+o(1),
  \label{eq:model-memory-scaling}\\
  \log_2 M_{\mathrm{state}}^{(n)}
  &=n\log_2 r,
  \label{eq:state-memory-scaling}
\end{align}
so the exact terminal model-recovery record rate is zero while the exact terminal record rate for deferred arbitrary-state recovery is \(\log_2r\) bits per use.
\end{theorem}

A concise proof sketch identifies the two resources. Choose a faithful signal state $\rho_z=(I+zZ)/2$, unknown interior probabilities $\boldsymbol q$, and fine errors $E_a=\sqrt{q_a}Xe^{-i\phi_aZ/2}$ with distinct phases. Every normalized conditional output is the same $\sigma_z=(I-zZ)/2$, while the probability-score vector is injective on the type counts. The type record therefore preserves the QFI and reconstructs the complete flagged family by uniformly expanding each type into its ordered trajectories. Its minimality against stochastic terminal processing follows because the ordered-trajectory probability vectors span the full type subspace. In contrast, distinct trajectories are pairwise Knill--Laflamme incompatible on the full $n$-qubit code~\cite{knill1997theory,nielsen2010quantum}, so each recoverable terminal output can support only one trajectory. The complete proof is in the Supplemental Material \cite{supp}.

\paragraph*{Growing alphabets.}
\phantomsection\label{cor:growing-alphabet-rate}
If the one-use alphabet has size $r_n$, the type-record rate obeys
\begin{equation}
  \lim_{n\to\infty}\frac1n\log_2\binom{n+r_n-1}{r_n-1}=0
  \quad\Longleftrightarrow\quad r_n=o(n).
  \label{eq:growing-alphabet-threshold}
\end{equation}
Thus zero terminal rate persists for every sublinear alphabet growth, not only for fixed $r$; the proof is in the Supplemental Material \cite{supp}.

The recovery cost has a simple graph formulation. Assume each fine error $E_a$ is individually correctable on a code projector $P$ and define the incompatibility graph $G_P$ by an edge $a\sim b$ whenever $PE_a^\dagger E_bP$ is not proportional to $P$.

\begin{theorem}[Deferred storage and online feedback]
\label{thm:recovery-resource}
The minimum one-use recovery alphabet is $\chi(G_P)$. For $n$ deferred uses of the product instrument, a terminal record must distinguish at least the chromatic requirements of the $n$-use incompatibility graph; for the complete-graph family of Theorems~\ref{thm:noncommuting-qfi-rate} and~\ref{thm:rate-separation}, this gives $r^n$ terminal records. If exact correction is instead applied after every use, then every reachable classical history requires at least $\chi(G_P)$ nonempty feedback symbols, and the tree of possible $n$-round control transcripts has at least $\chi(G_P)^n$ leaves. These bounds are attained by reusing a fixed optimal coloring and its class-dependent recoveries.
\end{theorem}

\begin{proof}
Every set of errors assigned to one terminal flag or one online control symbol must satisfy the Knill--Laflamme conditions and therefore forms an independent set of $G_P$. Covering all fine errors needs at least $\chi(G_P)$ such sets. Applying this argument at every reachable online history gives the per-use alphabet bound and, by induction, the transcript-tree bound. A proper coloring supplies matching recoveries and attains the bounds.
\end{proof}

Thus online correction can recycle a one-symbol register, but it does not eliminate syndrome information: terminal storage is replaced by repeated readout and conditioned control. If the noise weights $\boldsymbol q$ are known and only the signal $z$ is estimated in Theorem~\ref{thm:rate-separation}, all fine branches have the same signal SLD and one label suffices. The polynomial record is required for the complete joint signal-and-noise model, as in self-calibration, nuisance-parameter tracking, or syndrome-based noise spectroscopy \cite{wagner2022pauli}.

\section{Broader channel perspective and discussion}
\label{sec:discussion}

The finite-syndrome identity is the commutative special case of a general data-processing remainder. This broader perspective connects statistical sufficiency, multiplicative domains, and complementary-channel recovery~\cite{petz1988sufficiency,gao2024sufficient,choi1974schwarz,choi2009multiplicative,beny2009conditions,kretschmann2008information}. Let $\mathcal N$ be parameter independent, $\sigma_\theta=\mathcal N(\rho_\theta)$, $\Psi=\mathcal N^\dagger$, and $V$ a Stinespring isometry. For input and output SLDs $L$ and $T$:

\begin{theorem}[Stinespring Pythagorean identity]
\label{thm:stinespring-pythagorean}
Define the score-intertwining defect
\begin{equation}
  \mathcal E_{L,T}
  =[(T\otimes I_E)V-VL]\sqrt{\rho_\theta}.
  \label{eq:score-intertwining-defect}
\end{equation}
For every parameter-independent quantum channel,
\begin{align}
  F_Q(\rho_\theta)-F_Q(\sigma_\theta)
  &=\|\mathcal E_{L,T}\|_2^2
  \label{eq:stinespring-loss}\\
  &=\Tr\!\left\{\rho_\theta[L-\Psi(T)]^2\right\}\notag\\
  &\quad+\Tr\!\left\{\rho_\theta
  [\Psi(T^2)-\Psi(T)^2]\right\}.
  \label{eq:heisenberg-loss}
\end{align}
Both terms in Eq.~\eqref{eq:heisenberg-loss} are nonnegative.  The following conditions are equivalent:
\begin{align}
  F_Q(\rho_\theta)&=F_Q(\sigma_\theta),
  \label{eq:channel-zero-a}\\
  [L-\Psi(T)]\sqrt{\rho_\theta}&=0,
  \label{eq:channel-zero-b}\\
  [(T\otimes I_E)V-VL]\sqrt{\rho_\theta}&=0.
  \label{eq:channel-zero-c}
\end{align}
If \(\rho_\theta\) is faithful, these conditions are further equivalent to
\begin{equation}
  L=\Psi(T),
  \qquad
  \Psi(T^2)=\Psi(T)^2,
  \label{eq:faithful-channel-zero}
\end{equation}
so the output SLD \(T\) belongs to the multiplicative domain of \(\Psi\).
\end{theorem}

The proof, the multiparameter Gram form, and the exact randomization penalty for stochastic syndrome readout are given in the Supplemental Material \cite{supp}. The identity also locates the boundary between task-specific metrological preservation and full channel reversibility.

\paragraph*{Generating-score rigidity.}
\phantomsection\label{thm:score-algebra-rigidity}
Let $\rho_{\boldsymbol\theta}>0$ and suppose a parameter-independent channel preserves the QFI along a finite tuple of input SLDs $L_\mu$,
\begin{equation}
  J^{\mathrm{in}}_{\mu\mu}=J^{\mathrm{out}}_{\mu\mu}
  \quad\text{for every }\mu.
  \label{eq:rigidity-coordinate-equality}
\end{equation}
Writing $\mathcal A_T=C^*(I,T_1,\ldots,T_p)$ and $\mathcal A_L=C^*(I,L_1,\ldots,L_p)$, the channel adjoint restricts to a surjective unital $*$-homomorphism
\begin{equation}
  \mathcal N^\dagger\big|_{\mathcal A_T}:\mathcal A_T\to\mathcal A_L.
  \label{eq:score-homomorphism}
\end{equation}
If $\mathcal A_L=\mathcal L(\mathcal H)$, the channel is reversible on the entire input state space:
\begin{equation}
  \mathcal R\circ\mathcal N=\operatorname{id}_{\mathcal L(\mathcal H)}.
  \label{eq:full-channel-recovery}
\end{equation}
The proof is given in the Supplemental Material \cite{supp}.

A conditioned approximate version is useful only after algebraic generation is quantitatively stable. Define the completely bounded score-conditioning constant
\begin{equation}
  \kappa_{\mathrm{sc}}^{\mathrm{cb}}(\boldsymbol L)
  =\sup_{k\ge1}\sup_{K\notin M_k\otimes I}
  \frac{\|K-\mathcal P_k(K)\|_\infty}
  {\max_\mu\|[K,I_k\otimes L_\mu]\|_\infty},
  \label{eq:cb-score-conditioning}
\end{equation}
where $\mathcal P_k$ is the conditional expectation onto the amplified scalar commutant. If $\rho\ge\lambda I$ and the largest generating-coordinate QFI loss is $\delta$, the Supplemental Material \cite{supp} proves environmental leakage of order $\kappa_{\mathrm{sc}}^{\mathrm{cb}}\sqrt{\delta/\lambda}$ and a recovery error of order $\sqrt{\kappa_{\mathrm{sc}}^{\mathrm{cb}}}(\delta/\lambda)^{1/4}$. This is a structural stability result rather than a practical error budget; no dimension-independent bound exists without score conditioning.

The operational scope is deliberately narrower than arbitrary channel engineering. The orthogonal fine record must be physically available through an ancilla, monitored environment, erasure sector, or another genuine pointer degree of freedom. Direct coarse readout reduces the terminal alphabet but not the Hilbert-space dimension or accessibility cost of the fine sectors. The repeated-use theorems assume independent faithful blocks; correlated trajectories, rank-deficient order-forgetting criteria, and coherent-controller resource bounds require separate theories. Exact type compression concerns local SLD information, whereas complete model recovery and universal code recovery are progressively stronger tasks. Theorem~\ref{thm:full-model-recovery-obstruction} identifies a sharp obstruction to nonflat model recovery in the informationally complete faithful regime.

Finite syndrome resolution therefore creates a task-dependent information bottleneck. The exact loss geometry determines which branches can be merged; the common-score criterion determines when trajectory order is redundant; and the rate theorems show that metrological inference, complete statistical-model recovery, deferred state recovery, and online control can require fundamentally different forms of syndrome information. What must be remembered is determined jointly by the monitored instrument, the statistical task, the recovery target, and the timing of feedback.

\section*{Data Availability}
No experimental data were created or analyzed in this study. The numerical verification code supporting the analytic checks is available from the author upon reasonable request.

\bibliographystyle{apsrev4-2}
\bibliography{references}

\clearpage
\onecolumngrid
\begin{center}
  {\large\bfseries Supplemental Material for\\[0.35em]
  ``Task-Dependent Syndrome Memory in Quantum Sensing and State Recovery''}\\[0.9em]
  {\normalsize Jianqi Sheng}\\[0.25em]
  {\small Department of Physics, City University of Hong Kong, Hong Kong}
\end{center}
\vspace{0.8em}
\begin{center}
\begin{minipage}{0.88\textwidth}
\small
This Supplemental Material provides deferred proofs and structural results supporting the main text, including the exact flagged-QFI geometry, multiparameter channel identities, model-recovery constructions, and conditioned approximate rigidity. It also presents analytic qubit and continuous-axis examples, the exact extended-convexity gap, and technical details of the operator-clustering algorithm.
\end{minipage}
\end{center}
\vspace{1em}
\twocolumngrid

% Supplemental Material numbering.
\setcounter{section}{0}
\setcounter{subsection}{0}
\setcounter{equation}{0}
\setcounter{theorem}{0}
\setcounter{proposition}{0}
\setcounter{lemma}{0}
\setcounter{corollary}{0}
\setcounter{definition}{0}
\setcounter{remark}{0}
\setcounter{figure}{0}
\setcounter{table}{0}
\renewcommand{\thesection}{S\arabic{section}}
\renewcommand{\thesubsection}{S\arabic{section}.\Alph{subsection}}
\renewcommand{\theequation}{S\arabic{equation}}
\renewcommand{\thetheorem}{S\arabic{theorem}}
\renewcommand{\theproposition}{S\arabic{proposition}}
\renewcommand{\thelemma}{S\arabic{lemma}}
\renewcommand{\thecorollary}{S\arabic{corollary}}
\renewcommand{\thedefinition}{S\arabic{definition}}
\renewcommand{\theremark}{S\arabic{remark}}
\renewcommand{\thefigure}{S\arabic{figure}}
\renewcommand{\thetable}{S\arabic{table}}
\makeatletter
\@ifundefined{theHsection}{\newcommand{\theHsection}{S\arabic{section}}}{\renewcommand{\theHsection}{S\arabic{section}}}
\@ifundefined{theHsubsection}{\newcommand{\theHsubsection}{S\arabic{section}.\Alph{subsection}}}{\renewcommand{\theHsubsection}{S\arabic{section}.\Alph{subsection}}}
\@ifundefined{theHequation}{\newcommand{\theHequation}{S\arabic{equation}}}{\renewcommand{\theHequation}{S\arabic{equation}}}
\@ifundefined{theHtheorem}{\newcommand{\theHtheorem}{S\arabic{theorem}}}{\renewcommand{\theHtheorem}{S\arabic{theorem}}}
\@ifundefined{theHproposition}{\newcommand{\theHproposition}{S\arabic{proposition}}}{\renewcommand{\theHproposition}{S\arabic{proposition}}}
\@ifundefined{theHlemma}{\newcommand{\theHlemma}{S\arabic{lemma}}}{\renewcommand{\theHlemma}{S\arabic{lemma}}}
\@ifundefined{theHcorollary}{\newcommand{\theHcorollary}{S\arabic{corollary}}}{\renewcommand{\theHcorollary}{S\arabic{corollary}}}
\@ifundefined{theHdefinition}{\newcommand{\theHdefinition}{S\arabic{definition}}}{\renewcommand{\theHdefinition}{S\arabic{definition}}}
\@ifundefined{theHremark}{\newcommand{\theHremark}{S\arabic{remark}}}{\renewcommand{\theHremark}{S\arabic{remark}}}
\@ifundefined{theHfigure}{\newcommand{\theHfigure}{S\arabic{figure}}}{\renewcommand{\theHfigure}{S\arabic{figure}}}
\@ifundefined{theHtable}{\newcommand{\theHtable}{S\arabic{table}}}{\renewcommand{\theHtable}{S\arabic{table}}}
\makeatother

\section{Deferred derivations and structural results}
\label{sec:deferred-main-results}

This section collects derivations and auxiliary statements omitted from the main text to keep the rate-separation narrative compact.

\subsection{Flagged QFI and deterministic compression}

\begin{proposition}[QFI of an orthogonally flagged state]
\label{prop:flagged-qfi}
Let \(\Omega_\theta=\bigoplus_a q_a(\theta)\rho_{a,\theta}\), where all nonzero \(q_a\) are differentiable.  Then
\begin{equation}
  F_Q(\Omega_\theta)
  =\sum_a\frac{\dot q_a^2}{q_a}
  +\sum_a q_a F_Q(\rho_{a,\theta}).
  \label{eq:flagged-qfi}
\end{equation}
\end{proposition}

\begin{proof}
The SLD of a block-diagonal state is block diagonal.  Its \(a\)-th block is
\(
S_a=(\dot q_a/q_a)I+L_a
\).
Because \(\Tr(\rho_aL_a)=\Tr\dot\rho_a=0\),
\begin{align}
  F_Q(\Omega_\theta)
  &=\sum_a q_a\Tr\left[\rho_a
  \left(\frac{\dot q_a}{q_a}I+L_a\right)^2\right] \notag\\
  &=\sum_a\frac{\dot q_a^2}{q_a}
  +\sum_aq_a\Tr(\rho_aL_a^2).
\end{align}
\end{proof}

For a stochastic compressor $P(m|a)$, define $\lambda_f=\prod_aP(f(a)|a)$ over deterministic maps $f:\mathcal A\to[M]$. Then
\begin{equation}
 \sum_f\lambda_f=\prod_a\sum_mP(m|a)=1,
 \qquad
 \sum_{f:f(a)=m}\lambda_f=P(m|a),
\end{equation}
so the stochastic output is the corresponding convex combination of deterministic outputs. Convexity of the SLD QFI proves that randomization cannot improve the optimum.

\subsection{Pythagorean centers and the variational problem}

For a fixed class \(C_m\), define the state-weighted squared cost
\begin{equation}
  \mathcal C_m(T)=\sum_{a\in C_m}
  \Tr\left[\tau_a(S_a-T)^2\right],
  \qquad T=T^\dagger.
  \label{eq:class-cost}
\end{equation}

\begin{proposition}[Pythagorean SLD-center identity]
\label{prop:sld-center}
Let \(T_m\) be any SLD score of
\(\tau_m=\sum_{a\in C_m}\tau_a\).  For every Hermitian \(T\),
\begin{equation}
  \mathcal C_m(T)
  =\mathcal C_m(T_m)
  +\Tr\!\left[\tau_m(T-T_m)^2\right].
  \label{eq:center-pythagorean}
\end{equation}
Consequently, \(T_m\) is a global minimizer of \(\mathcal C_m\), and all minimizers are exactly the Hermitian operators satisfying
\begin{equation}
  (T-T_m)\sqrt{\tau_m}=0.
  \label{eq:center-minimizers}
\end{equation}
If \(\tau_m\) is faithful, the minimizer is unique.
\end{proposition}

\begin{proof}
Set \(H=T-T_m\).  Expanding \(S_a-T=(S_a-T_m)-H\) gives
\begin{align}
  \mathcal C_m(T)
  ={}&\mathcal C_m(T_m)+\Tr(\tau_mH^2)\notag\\
  &-\sum_{a\in C_m}\Tr\!\left\{\tau_a
  \left[(S_a-T_m)H+H(S_a-T_m)\right]\right\}.
  \label{eq:center-expand}
\end{align}
The cross term vanishes.  Indeed,
\begin{align}
  \sum_{a\in C_m}\Tr[\tau_a(S_aH+HS_a)]
  &=2\sum_{a\in C_m}\Tr(\dot\tau_aH)\notag\\
  &=2\Tr(\dot\tau_mH),
\end{align}
whereas the coarse SLD equation gives
\begin{equation}
  \Tr[\tau_m(T_mH+HT_m)]=2\Tr(\dot\tau_mH).
\end{equation}
This proves Eq.~\eqref{eq:center-pythagorean}.  Its final term is
\begin{equation}
  \Tr[\tau_mH^2]=\|H\sqrt{\tau_m}\|_2^2\geq0,
\end{equation}
with equality exactly when Eq.~\eqref{eq:center-minimizers} holds.  Faithfulness of \(\tau_m\) then implies \(H=0\).
\end{proof}

Combining Theorem~1 of the main text and Proposition~\ref{prop:sld-center} gives the variational formulation.

\begin{theorem}[Finite-syndrome variational principle]
\label{thm:kmeans}
Define
\begin{equation}
  F_M^\star(\theta)
  =\max_{f:\mathcal A\to[M]}F_Q(\Omega_\theta^{(f)}).
  \label{eq:FM-def}
\end{equation}
Then
\begin{equation}
  F_M^\star(\theta)
  =F_Q(\Omega_\theta^{\mathrm{fine}})-D_M(\theta),
  \label{eq:FM-loss}
\end{equation}
where
\begin{equation}
  D_M(\theta)=
  \min_{\substack{f:\mathcal A\to[M]\\
  T_1=T_1^\dagger,\ldots,T_M=T_M^\dagger}}
  \sum_a\Tr\left[
    \tau_a(S_a-T_{f(a)})^2
  \right].
  \label{eq:operator-kmeans}
\end{equation}
\end{theorem}

The quantity \(F_M^\star(\theta)\) is a local design-point optimum.  A physical compressor intended to work over an interval \(\Theta\) must use one parameter-independent partition.  For a prior density \(\pi(\theta)\), an operational Bayesian objective is
\begin{equation}
  D_{M,\pi}
  =\min_f\int_\Theta \pi(\theta)
  \left[F_Q(\Omega_\theta^{\mathrm{fine}})
  -F_Q(\Omega_\theta^{(f)})\right]d\theta,
  \label{eq:bayesian-compression}
\end{equation}
while a robust alternative is
\begin{equation}
  D_{M,\mathrm{wc}}
  =\min_f\sup_{\theta\in\Theta}
  \left[F_Q(\Omega_\theta^{\mathrm{fine}})
  -F_Q(\Omega_\theta^{(f)})\right].
  \label{eq:worst-case-compression}
\end{equation}
In both cases the same partition \(f\) is used for every \(\theta\).  A parameter-dependent partition requires a separately specified adaptive-estimation protocol.

\subsection{Multiparameter channel identity and conditioned rigidity}

\begin{theorem}[Stinespring Gram identity for the QFI matrix]
\label{thm:stinespring-qfim}
For a multiparameter model, let \(L_\mu\) and \(T_\mu\) be the input and output SLDs and define
\begin{equation}
  X_\mu=
  [(T_\mu\otimes I_E)V-VL_\mu]\sqrt{\rho_\theta}.
  \label{eq:channel-defect-vectors}
\end{equation}
Then
\begin{equation}
  J^{\mathrm{in}}_{\mu\nu}-J^{\mathrm{out}}_{\mu\nu}
  =\operatorname{Re}\Tr(X_\mu^\dagger X_\nu).
  \label{eq:stinespring-qfim-gram}
\end{equation}
Hence the QFI-matrix loss is positive semidefinite.  The full matrix is preserved if and only if \(X_\mu=0\) for every \(\mu\), equivalently
\begin{equation}
  [L_\mu-\Psi(T_\mu)]\sqrt{\rho_\theta}=0
  \quad\text{for every }\mu.
  \label{eq:general-channel-multiparameter-zero}
\end{equation}
\end{theorem}

\begin{proof}
For every real vector \(\boldsymbol v\), apply Theorem~5 of the main text to the one-dimensional submodel generated by \(\partial_{\boldsymbol v}=\sum_\mu v_\mu\partial_\mu\).  Its Stinespring defect is \(X_{\boldsymbol v}=\sum_\mu v_\mu X_\mu\), so
\begin{equation}
  \boldsymbol v^{\mathsf T}(J^{\mathrm{in}}-J^{\mathrm{out}})\boldsymbol v
  =\left\|\sum_\mu v_\mu X_\mu\right\|_2^2.
\end{equation}
Polarization gives Eq.~\eqref{eq:stinespring-qfim-gram}.  Vanishing of the full Gram matrix is equivalent to vanishing of every diagonal norm \(\|X_\mu\|_2^2\).  The final equivalence follows from the scalar zero-loss conditions in Theorem~5 of the main text.
\end{proof}

The exact statement admits a quantitative extension only after the algebraic generation condition is made stable.  Let \(d=\dim\mathcal H_A\), let \(M_n\) denote the \(n\times n\) matrices, and define the conditional expectation onto the amplified scalar commutant by
\begin{equation}
  \mathcal P_n(K)
  =\left[(\operatorname{id}_n\otimes\Tr/d)(K)\right]\otimes I_d,
  \qquad K\in M_n\otimes M_d.
  \label{eq:scalar-conditional-expectation}
\end{equation}

\begin{definition}[Completely bounded score-conditioning constant]
\label{def:cb-score-conditioning}
For a Hermitian tuple \(\boldsymbol L=(L_1,\ldots,L_p)\) satisfying
\(C^*(I,L_1,\ldots,L_p)=M_d\), define
\begin{equation}
  \kappa_{\mathrm{sc}}^{\mathrm{cb}}(\boldsymbol L)
  =\sup_{n\geq1}\ 
   \sup_{K\notin M_n\otimes I_d}
  \frac{\|K-\mathcal P_n(K)\|_\infty}
  {\max_\mu\|[K,I_n\otimes L_\mu]\|_\infty}.
  \label{eq:supp-cb-score-conditioning}
\end{equation}
The denominator is nonzero outside \(M_n\otimes I_d\), and
\(\kappa_{\mathrm{sc}}^{\mathrm{cb}}(\boldsymbol L)<\infty\); a direct finite-dimensional proof is given in Sec.~\ref{sec:conditioned-rigidity}.
\end{definition}

This constant is invariant under a common rescaling in the combination appearing below: rescaling all scores by \(c\neq0\) divides \(\kappa_{\mathrm{sc}}^{\mathrm{cb}}\) by \(|c|\) and multiplies the coordinate QFI losses by \(c^2\).

\begin{theorem}[Conditioned approximate score rigidity]
\label{thm:conditioned-approximate-rigidity}
Let \(\rho\geq\lambda I_d\) with \(\lambda>0\), let
\(\mathcal N:M_d\to\mathcal L(\mathcal H_B)\) be a channel, and let the input SLD tuple \(\boldsymbol L=(L_1,\ldots,L_p)\) generate \(M_d\).  For the corresponding output SLDs \(T_\mu\), set
\begin{equation}
  \delta_\mu
  =J^{\mathrm{in}}_{\mu\mu}-J^{\mathrm{out}}_{\mu\mu},
  \qquad
  \delta=\max_\mu\delta_\mu.
  \label{eq:coordinate-qfi-loss-delta}
\end{equation}
Let \(\widehat{\mathcal N}\) be a complementary channel,
\(\omega_E=\widehat{\mathcal N}(I_d/d)\), and let
\(\mathcal S_{\omega_E}(X)=\Tr(X)\omega_E\) be the corresponding replacer channel.  Then
\begin{equation}
  \left\|\widehat{\mathcal N}-\mathcal S_{\omega_E}\right\|_\diamond
  \leq
  2\kappa_{\mathrm{sc}}^{\mathrm{cb}}(\boldsymbol L)
  \sqrt{\frac{\delta}{\lambda}}.
  \label{eq:environment-leakage-bound}
\end{equation}
Consequently, there exists a recovery channel \(\mathcal R\) such that
\begin{equation}
  \left\|\mathcal R\circ\mathcal N-
  \operatorname{id}_{M_d}\right\|_\diamond
  \leq
  \min\left\{2,
  2\sqrt{2\kappa_{\mathrm{sc}}^{\mathrm{cb}}(\boldsymbol L)}
  \left(\frac{\delta}{\lambda}\right)^{1/4}
  \right\}.
  \label{eq:conditioned-diamond-recovery}
\end{equation}
\end{theorem}

Theorem~\ref{thm:conditioned-approximate-rigidity} is not a general sufficiency theorem for the SLD metric.  Its additional input is a finite tuple of SLDs that generates the full algebra with quantitative amplified conditioning.  Prior approximate-algebra criteria require control of the complementary channel against the full amplified commutant~\cite{beny2009conditions}; Eqs.~\eqref{eq:coordinate-qfi-loss-delta}--\eqref{eq:environment-leakage-bound} derive that control from finitely many local metrological losses.  The final passage from approximate environmental constancy to recovery is the channel information--disturbance theorem~\cite{kretschmann2008information}.  The proof is given in Sec.~\ref{sec:conditioned-rigidity}.

\subsection{Noncommuting recovery on a sufficient subsystem}

\begin{proposition}[Noncommuting full-model recovery on a sufficient subsystem]
\label{prop:noncommuting-sufficient-subsystem}
There exists a faithful two-label classical--quantum family whose normalized conditional states are noncommuting, but for which forgetting the label is exactly reversible with no retained classical record.  The parameterized coarse family occupies a proper sufficient subsystem and therefore does not satisfy Eq.~(30) of the main text.
\end{proposition}

A constructive qubit--ancilla realization is given in Sec.~\ref{sec:noncommuting-subsystem}.  Proposition~\ref{prop:noncommuting-sufficient-subsystem} shows that the informational-completeness hypothesis of the complete-model obstruction in Sec.~IV of the main text is essential.  It does not, however, produce the same recovery-rate separation as Theorem~3 of the main text, because its branch operators are not individually correctable on the full input space.  Whether faithful noncommuting full-model recovery and exponential universal-recovery separation can coexist in one individually correctable instrument remains open.

\subsection{Proofs of main-text theorems}

\begin{proof}[Proof of the orthogonal-pointer realization]
Equation~(4) of the main text gives
$V^\dagger V=\sum_aE_a^\dagger E_a=I$, so the map $V$ defined in Eq.~(5) of the main text is an isometry.  Expanding $V\rho_\theta V^\dagger$ and dephasing, or equivalently measuring, the ancilla in the pointer basis removes the off-diagonal terms $|a\rangle\langle b|$ and yields Eq.~(6) of the main text.  Independent repetitions use the tensor-product isometry $V^{\otimes n}$ and fresh pointer registers, which gives Eq.~(7) of the main text.  A subsequent classical channel on the pointer basis acts trivially on the tensor-product quantum output and hence realizes exactly the allowed record compression.
\end{proof}

\begin{proof}[Proof of stable type compression]
For a real direction \(\boldsymbol v\), write
\(L_{\boldsymbol v}=\sum_\mu v_\mu L_\mu\),
\(s_{a,\boldsymbol v}=\sum_\mu v_\mu s_{a,\mu}\), and
\(\Delta_{a,\boldsymbol v}=\sum_\mu v_\mu\Delta_{a,\mu}\).
For a trajectory \(\boldsymbol a\) of type
\(\boldsymbol k=(k_1,\ldots,k_r)\), its directional product score is
\begin{equation}
  S_{\boldsymbol a,\boldsymbol v}^{(n)}
  =\sum_{i=1}^n L_{\boldsymbol v}^{(i)}
  +\left(\sum_a k_as_{a,\boldsymbol v}\right)I
  +\sum_{i=1}^n\Delta_{a_i,\boldsymbol v}^{(i)}.
  \label{eq:approximate-product-score}
\end{equation}
The first two terms depend only on the type and hence form an admissible trial center for that type class.  The Pythagorean center identity therefore bounds the directional loss by
\begin{equation}
  \sum_{\boldsymbol a\in\mathcal A^n}
  \Tr\!\left[
  \tau_{\boldsymbol a}^{(n)}
  \left(\sum_{i=1}^n
  \Delta_{a_i,\boldsymbol v}^{(i)}\right)^2
  \right].
  \label{eq:trial-type-loss}
\end{equation}
All cross terms with \(i\neq j\) vanish after summing over trajectories, because the tensor-product trace factorizes and
\(\Tr(\tau_a\Delta_{a,\boldsymbol v})=0\) for every \(a\).  The diagonal terms give
\begin{equation}
  n\sum_a\Tr(\tau_a\Delta_{a,\boldsymbol v}^2)
  =n\,\boldsymbol v^{\mathsf T}\Gamma\boldsymbol v.
\end{equation}
The exact QFI loss matrix is positive semidefinite, and the preceding inequality holds for every real \(\boldsymbol v\); this proves the matrix inequality in Eq.~(25) of the main text.
\end{proof}

\begin{proof}[Proof of the noncommuting example]
Choose a faithful qubit seed
\begin{align}
  \rho_0&=\frac12(I+xX+z_0Z),\\
  x&\neq0,
  \qquad x^2+z_0^2<1.
\end{align}
and define the one-parameter SLD exponential family
\begin{equation}
  \rho_\theta
  =\frac{e^{\theta Z/2}\rho_0e^{\theta Z/2}}
  {\Tr(e^{\theta Z/2}\rho_0e^{\theta Z/2})}.
  \label{eq:noncommuting-input-family}
\end{equation}
Writing \(\zeta(\theta)=\Tr(Z\rho_\theta)\), direct differentiation gives
\begin{equation}
  \partial_\theta\rho_\theta
  =\frac12\{Z-\zeta(\theta)I,\rho_\theta\},
  \label{eq:noncommuting-input-sld}
\end{equation}
so its SLD is \(L_\theta=Z-\zeta(\theta)I\).

Let \(\boldsymbol q\) be an unknown interior probability vector and choose phases \(\phi_a\) inside an interval of length strictly smaller than \(\pi\), so that they are pairwise distinct and no two differ by \(\pi\).  Define
\begin{align}
  W_a&=e^{-i\phi_aZ/2},
  &E_a&=\sqrt{q_a}W_a,\notag\\
  \tau_{a,\theta,\boldsymbol q}
  &=q_aW_a\rho_\theta W_a^\dagger.
  \label{eq:noncommuting-errors}
\end{align}
Since every \(W_a\) commutes with \(Z\), the signal score is branch independent,
\begin{equation}
  S_{a,\theta}=Z-\zeta(\theta)I,
\end{equation}
For the independent coordinates \(q_1,\ldots,q_{r-1}\), with \(q_r=1-\sum_{\mu=1}^{r-1}q_\mu\), the remaining scores are scalar:
\begin{equation}
  S_{a,q_\mu}=
  \begin{cases}
    q_\mu^{-1}I, & a=\mu,\\
    -q_r^{-1}I, & a=r,\\
    0, & \text{otherwise}.
  \end{cases}
  \label{eq:noncommuting-probability-sld}
\end{equation}
Hence the common-score form holds.  If a trajectory has type \(\boldsymbol k=(k_1,\ldots,k_r)\), its accumulated score for \(q_\mu\) is
\begin{equation}
  u_\mu(\boldsymbol k)=\frac{k_\mu}{q_\mu}-\frac{k_r}{q_r}.
  \label{eq:noncommuting-type-score}
\end{equation}
If two types \(\boldsymbol k,\boldsymbol\ell\) give the same score vector and \(\Delta_a=k_a-\ell_a\), then
\(\Delta_\mu/q_\mu=\Delta_r/q_r=c\) for every \(\mu<r\).  Since \(\sum_a\Delta_a=0\) and \(\sum_aq_a=1\), one has \(c=0\) and therefore \(\boldsymbol k=\boldsymbol\ell\).  The accumulated-score statement in Theorem~2 of the main text gives the first equality in Eq.~(27) of the main text.

The conditional Bloch vectors are obtained from that of \(\rho_\theta\) by distinct rotations about the \(Z\) axis.  Their transverse component is nonzero because \(x\neq0\); the chosen phase differences therefore make the normalized states \(W_a\rho_\theta W_a^\dagger\) pairwise noncommuting.  For arbitrary-state recovery,
\begin{equation}
  E_a^\dagger E_b=\sqrt{q_aq_b}
  e^{-i(\phi_b-\phi_a)Z/2}
\end{equation}
is proportional to the identity only when \(a=b\).  For two \(n\)-use trajectories, their cross operator is the tensor product of these one-use cross operators.  If that tensor product were scalar, comparing computational-basis eigenvalues that differ in only one tensor factor would force the corresponding phase difference to vanish modulo \(2\pi\).  Hence the trajectories must agree at every position.  The \(n\)-use incompatibility graph is therefore complete, so exact deferred recovery requires all \(r^n\) trajectories.
\end{proof}

\begin{proof}[Proof of the informationally complete obstruction]
Let \(\mathcal C\) be the forgetful channel.  Exact recovery implies equality in the relative-entropy data-processing inequality between every family member and the faithful reference state.  The Petz recovery map for \((\mathcal C,\Omega_{\boldsymbol\vartheta_0})\) therefore also recovers the family~\cite{petz1988sufficiency}.  Since \(\mathcal C^\dagger(Y)=\bigoplus_jY\), its \(j\)th output block is
\begin{equation}
  \tau_{j,\boldsymbol\vartheta_0}^{1/2}
  X_0^{-1/2}X_{\boldsymbol\vartheta}X_0^{-1/2}
  \tau_{j,\boldsymbol\vartheta_0}^{1/2},
\end{equation}
which proves Eq.~(28) of the main text.  Equation~(29) of the main text follows by summing the reference blocks.  Conversely, if Eq.~(28) of the main text holds, then
\begin{equation}
  \mathcal R(Y)=\bigoplus_j A_jYA_j^\dagger
\end{equation}
is a channel by Eq.~(29) of the main text and recovers every \(\Omega_{\boldsymbol\vartheta}\).

Define the channel
\(\mathcal E(Y)=\sum_jA_jYA_j^\dagger\).
Equation~(28) of the main text gives
\(\mathcal E(X_{\boldsymbol\vartheta})=X_{\boldsymbol\vartheta}\) for every parameter value.  Under Eq.~(30) of the main text, linearity forces \(\mathcal E\) to be the identity channel on \(\mathcal L(\mathcal H)\).  The identity channel has Choi rank one, so every Kraus operator in any Kraus representation is proportional to the identity: \(A_j=c_jI\).  Equation~(28) of the main text then gives Eq.~(31) of the main text with \(p_j=|c_j|^2\).
\end{proof}

\begin{proof}[Proof of the exact rate separation]
Take
\begin{equation}
  \rho_z=\frac12(I+zZ),
  \qquad |z|<1,
  \label{eq:rate-separation-input}
\end{equation}
and let \(\boldsymbol q=(q_1,\ldots,q_r)\) lie in the interior of the probability simplex.  Use independent coordinates \((q_1,\ldots,q_{r-1})\), with
\(q_r=1-\sum_{\mu=1}^{r-1}q_\mu\).  Choose phases \(\phi_a\) that are pairwise distinct modulo \(2\pi\), and define
\begin{equation}
  W_a=e^{-i\phi_aZ/2},
  \qquad
  U_a=XW_a,
  \qquad
  E_a=\sqrt{q_a}\,U_a.
  \label{eq:rate-separation-errors}
\end{equation}
The fine blocks are
\begin{equation}
  \tau_{a,z,\boldsymbol q}
  =E_a\rho_zE_a^\dagger
  =q_a\sigma_z,
  \qquad
  \sigma_z=\frac12(I-zZ),
  \label{eq:rate-separation-branches}
\end{equation}
because every \(W_a\) commutes with \(\rho_z\).  The one-use SLDs are
\begin{align}
  S_{a,z}&=L_z=-\frac{zI+Z}{1-z^2},
  \label{eq:rate-signal-sld}\\
  S_{a,q_\mu}
  &=
  \begin{cases}
    q_\mu^{-1}I, & a=\mu,\\
    -q_r^{-1}I, & a=r,\\
    0, & \text{otherwise},
  \end{cases}
  \qquad \mu=1,\ldots,r-1.
  \label{eq:rate-probability-sld}
\end{align}
Thus the accumulated-score statement in Theorem~2 of the main text applies.  For a trajectory of type
\(\boldsymbol k=(k_1,\ldots,k_r)\), its accumulated scalar score for parameter \(q_\mu\) is
\begin{equation}
  u_\mu(\boldsymbol k)
  =\frac{k_\mu}{q_\mu}-\frac{k_r}{q_r}.
  \label{eq:type-score}
\end{equation}
Suppose two types \(\boldsymbol k\) and \(\boldsymbol\ell\) give the same score vector, and set \(\Delta_a=k_a-\ell_a\).  Equation~\eqref{eq:type-score} gives
\begin{equation}
  \frac{\Delta_\mu}{q_\mu}
  =\frac{\Delta_r}{q_r}
  \equiv c,
  \qquad \mu=1,\ldots,r-1.
  \label{eq:type-injectivity-step}
\end{equation}
Since both types sum to \(n\),
\(0=\sum_a\Delta_a=c\sum_aq_a=c\), and hence every \(\Delta_a=0\).  The score map is therefore injective on the type vectors.  There are
\(\binom{n+r-1}{r-1}\) types, so the accumulated-score statement in Theorem~2 of the main text proves the first equality for \(M_{\mathrm{QFI}}^{(n)}\).

It remains to show that the type also recovers the complete fine model.  Let
\begin{equation}
  N_{\boldsymbol k}=\frac{n!}{\prod_{a=1}^r k_a!}
  \label{eq:type-class-size}
\end{equation}
be the number of trajectories of type \(\boldsymbol k\).  Their coarse type block is
\begin{equation}
  \tau_{\boldsymbol k,z,\boldsymbol q}^{(n)}
  =N_{\boldsymbol k}
  \left(\prod_{a=1}^r q_a^{k_a}\right)
  \sigma_z^{\otimes n}.
  \label{eq:type-coarse-block}
\end{equation}
A parameter-independent recovery map replaces the type label by a uniformly random ordered trajectory of that type and leaves the quantum system unchanged:
\begin{equation}
  |\boldsymbol k\rangle\!\langle\boldsymbol k|\otimes A
  \longmapsto
  \frac{1}{N_{\boldsymbol k}}
  \sum_{\boldsymbol a:\,\operatorname{type}(\boldsymbol a)=\boldsymbol k}
  |\boldsymbol a\rangle\!\langle\boldsymbol a|\otimes A.
  \label{eq:type-recovery-map}
\end{equation}
Applied to Eq.~\eqref{eq:type-coarse-block}, it exactly reconstructs every fine block for all \((z,\boldsymbol q)\).  Hence
\(M_{\mathrm{model}}^{(n)}\leq\binom{n+r-1}{r-1}\).

This upper bound remains minimal even if the terminal classical compressor is allowed to be stochastic.  Let \(\boldsymbol p(\boldsymbol q)\in\mathbb R^{r^n}\) be the ordered-trajectory probability vector, with
\(p_{\boldsymbol a}(\boldsymbol q)=\prod_i q_{a_i}\).  If a classical channel \(C\) with \(M\) outputs followed by a recovery channel \(R\) reconstructs the fine flagged family, then tracing out the unchanged quantum output gives
\begin{equation}
  \boldsymbol p(\boldsymbol q)=RC\boldsymbol p(\boldsymbol q)
  \qquad\text{for every interior }\boldsymbol q.
  \label{eq:classical-model-recovery}
\end{equation}
Equation~\eqref{eq:classical-model-recovery} makes \(C\) injective on the span of the fine-model probability vectors.  Since the range of \(C\) is contained in \(\mathbb R^M\),
\begin{equation}
  M\geq\operatorname{rank}C
  \geq
  \dim\operatorname{span}
  \{\boldsymbol p(\boldsymbol q):\boldsymbol q\}.
  \label{eq:model-rank-lower-bound}
\end{equation}  This span has dimension
\(\binom{n+r-1}{r-1}\).  Indeed,
\begin{equation}
  \boldsymbol p(\boldsymbol q)
  =\sum_{\boldsymbol k}
  \boldsymbol 1_{\boldsymbol k}
  \prod_{a=1}^r q_a^{k_a},
  \label{eq:type-vector-expansion}
\end{equation}
where the vectors \(\boldsymbol 1_{\boldsymbol k}\), supported on distinct type classes, are linearly independent.  The homogeneous monomials
\(\prod_aq_a^{k_a}\) are also linearly independent on the probability-simplex interior.  Indeed, a homogeneous polynomial of degree \(n\) that vanished there would, by positive rescaling, vanish throughout the open positive orthant and hence be the zero polynomial.  Thus the coefficient vectors in Eq.~\eqref{eq:type-vector-expansion} span the full type subspace.  Therefore
\(M_{\mathrm{model}}^{(n)}\geq\binom{n+r-1}{r-1}\), proving the second equality in Eq.~(32) of the main text for arbitrary classical terminal post-processing.

For arbitrary-state recovery, the \(n\)-use fine errors are
\begin{equation}
  E_{\boldsymbol a}^{(n)}
  =E_{a_1}\otimes\cdots\otimes E_{a_n}.
  \label{eq:n-copy-errors}
\end{equation}
Every error is individually correctable on the full \(n\)-qubit code.  For two trajectories,
\begin{equation}
  \left(E_{\boldsymbol a}^{(n)}\right)^\dagger
  E_{\boldsymbol b}^{(n)}
  =\sqrt{q_{\boldsymbol a}q_{\boldsymbol b}}
  \bigotimes_{i=1}^n
  e^{-i(\phi_{b_i}-\phi_{a_i})Z/2},
  \label{eq:n-copy-cross-error}
\end{equation}
where \(q_{\boldsymbol a}=\prod_iq_{a_i}\).  If this operator were proportional to the identity, comparing two computational-basis eigenvalues that differ only in the \(i\)th bit would give
\(e^{-i(\phi_{b_i}-\phi_{a_i})}=1\).  Pairwise distinct phases then imply \(a_i=b_i\) for every \(i\).  Thus every two distinct trajectories are adjacent in the Knill--Laflamme incompatibility graph, which is \(K_{r^n}\).  The conclusion also survives arbitrary stochastic post-processing of the classical record.  If \(C(m|\boldsymbol a)>0\), then conditioned on output \(m\) the corresponding error channel has Kraus operator
\(\sqrt{C(m|\boldsymbol a)}E_{\boldsymbol a}^{(n)}\).  The Knill--Laflamme condition therefore forbids one output value from having positive support on two distinct trajectories.  Since every trajectory must occur in the support of at least one output, any recoverable terminal alphabet has at least \(r^n\) values.  Keeping the full trajectory attains this bound, and hence
\(M_{\mathrm{state}}^{(n)}=r^n\).

Finally, for fixed \(r\), the standard expansion of the binomial coefficient gives Eq.~(33) of the main text, while Eq.~(34) of the main text is immediate.
\end{proof}

\begin{proof}[Proof of the growing-alphabet threshold]
Set \(k_n=r_n-1\).  If \(k_n=o(n)\), the standard bound
\begin{equation}
  \binom{n+k_n}{k_n}
  \leq\left[\frac{e(n+k_n)}{k_n}\right]^{k_n}
\end{equation}
gives
\begin{equation}
  \frac1n\log T_n
  \leq\frac{k_n}{n}
  \left[1+\log\left(1+\frac{n}{k_n}\right)\right]\longrightarrow0.
\end{equation}
Conversely, if \(r_n\neq o(n)\), there are \(\varepsilon>0\) and a subsequence with \(k_n\geq\varepsilon n\).  With \(j_n=\lfloor\varepsilon n\rfloor\), monotonicity and the product formula for binomial coefficients give
\begin{equation}
  T_n\geq\binom{n+j_n}{j_n}
  \geq\left(1+\frac{n}{j_n}\right)^{j_n}.
\end{equation}
The logarithm divided by \(n\) is then bounded away from zero along that subsequence, proving Eq.~(35) of the main text.
\end{proof}

\begin{proof}[Proof of the Stinespring identity]
The input and output SLD equations give the cross identity
\begin{align}
  \Tr\{\rho_\theta[L\Psi(T)+\Psi(T)L]\}
  &=2\Tr[\dot\rho_\theta\Psi(T)]\notag\\
  &=2\Tr[\dot\sigma_\theta T]
  =2\Tr(\sigma_\theta T^2).
  \label{eq:channel-cross-identity}
\end{align}
Expanding the right-hand side of Eq.~(38) of the main text, using
\(\Tr[\rho_\theta\Psi(T^2)]=\Tr(\sigma_\theta T^2)\), and then applying Eq.~\eqref{eq:channel-cross-identity} gives
\(F_Q(\rho_\theta)-F_Q(\sigma_\theta)\).
The first term in Eq.~(38) of the main text is a squared Hilbert--Schmidt norm.  The second is nonnegative because the Kadison--Schwarz inequality for the unital completely positive map \(\Psi\) gives
\begin{equation}
  \Psi(T^2)-\Psi(T)^2\succeq0.
\end{equation}

For the Stinespring form, note that
\begin{equation}
  \Psi(Y)=V^\dagger(Y\otimes I_E)V.
\end{equation}
Direct expansion of the norm in Eq.~(37) of the main text, followed by Eq.~\eqref{eq:channel-cross-identity}, yields the same QFI difference.  Equivalently, with \(P_V=VV^\dagger\),
\begin{align}
  &[(T\otimes I_E)V-VL]\sqrt{\rho_\theta}\notag\\
  &=(I-P_V)(T\otimes I_E)V\sqrt{\rho_\theta}
  +V[\Psi(T)-L]\sqrt{\rho_\theta},
  \label{eq:stinespring-orthogonal-split}
\end{align}
whose two terms have orthogonal ranges.  Their squared norms are exactly the two terms in Eq.~(38) of the main text.

Equations~(39) and~(41) of the main text are equivalent by Eq.~(37) of the main text.  Equation~(39) of the main text implies Eq.~(40) of the main text because both terms in Eq.~(38) of the main text are nonnegative.  Conversely, Eq.~(40) of the main text makes the first term vanish and gives
\(F_Q(\rho_\theta)=\Tr[\rho_\theta\Psi(T)^2]\).  Substitution into Eq.~\eqref{eq:channel-cross-identity} yields
\(F_Q(\rho_\theta)=F_Q(\sigma_\theta)\), proving the reverse implication.  If \(\rho_\theta>0\), zero expectation of either positive operator in Eq.~(38) of the main text implies that operator is zero.  Thus Eq.~(42) of the main text follows.  For a Hermitian element, equality in Kadison--Schwarz is precisely the multiplicative-domain condition \cite{choi1974schwarz,choi2009multiplicative}.
\end{proof}

\begin{proof}[Proof of generating-score rigidity]
Equation~(43) of the main text makes every listed scalar QFI loss vanish.  Since \(\rho_{\boldsymbol\theta}\) is faithful, Theorem~5 of the main text implies, for every \(\mu\),
\begin{equation}
  L_\mu=\Psi(T_\mu),
  \qquad
  \Psi(T_\mu^2)=\Psi(T_\mu)^2.
  \label{eq:rigidity-generator-equalities}
\end{equation}
Thus every \(T_\mu\) lies in the multiplicative domain \(\operatorname{MD}(\Psi)\).  The multiplicative domain is a unital \(C^*\)-subalgebra, and \(\Psi\) acts as a \(*\)-homomorphism on it \cite{choi1974schwarz,choi2009multiplicative}.  Hence \(\mathcal A_T\subseteq\operatorname{MD}(\Psi)\), and the restriction in Eq.~(44) of the main text is a unital \(*\)-homomorphism.  Its image contains every \(L_\mu\) and therefore equals \(\mathcal A_L\).

Now assume \(\mathcal A_L=\mathcal L(\mathcal H_A)\), and set
\(\pi=\Psi|_{\mathcal A_T}\).  The kernel of a \(*\)-homomorphism between finite-dimensional \(C^*\)-algebras is a two-sided ideal generated by a central projection.  Consequently, there is a central projection \(e\in\mathcal A_T\) such that
\begin{equation}
  \mathcal A_T=\ker\pi\oplus e\mathcal A_T,
\end{equation}
with \(\pi|_{e\mathcal A_T}\) a \(*\)-isomorphism onto
\(\mathcal L(\mathcal H_A)\).  Let
\begin{equation}
  \jmath=(\pi|_{e\mathcal A_T})^{-1}
\end{equation}
be its completely positive inverse, so \(\jmath(I_A)=e\).  Choose any state \(\varphi\) on \(\mathcal L(\mathcal H_A)\) and define
\begin{equation}
  \mathcal R^\dagger(X)
  =\jmath(X)+\varphi(X)(I_B-e).
  \label{eq:recovery-adjoint-construction}
\end{equation}
This map is completely positive and unital.  Moreover, \(\Psi(e)=I_A\), whence \(\Psi(I_B-e)=0\), and therefore
\begin{equation}
  \Psi\circ\mathcal R^\dagger=\operatorname{id}_{\mathcal L(\mathcal H_A)}.
\end{equation}
Taking adjoints gives Eq.~(45) of the main text.
\end{proof}

\section{Qubit random-unitary syndrome compression}
%=========================================================================%

Let
\begin{equation}
  \rho_\theta=\frac{1}{2}\left(I+\boldsymbol r_\theta\cdot\boldsymbol\sigma\right)
\end{equation}
be a qubit state, and let a parameter-independent random-unitary channel have probabilities \(q_a\) and unitaries \(U_a\).  Conjugation by \(U_a\) induces a rotation \(R_a\in\mathrm{SO}(3)\) of the Bloch vector.  For a class \(C_m\), define
\begin{equation}
  p_m=\sum_{a\in C_m}q_a,
  \qquad
  A_m=\frac{1}{p_m}\sum_{a\in C_m}q_aR_a.
  \label{eq:Am-def}
\end{equation}
The normalized conditional state has Bloch vector
\begin{equation}
  \boldsymbol s_m=A_m\boldsymbol r_\theta,
  \qquad
  \dot{\boldsymbol s}_m=A_m\dot{\boldsymbol r}_\theta.
\end{equation}

\begin{proposition}[Finite partition formula for qubit random-unitary noise]
\label{prop:qubit-partition}
For parameter-independent \(q_a\), the QFI retained by a partition \(f\) is
\begin{equation}
\begin{split}
  F_Q(\Omega_\theta^{(f)})
  =\sum_{m:p_m>0}p_m\Bigg[
  &\left|A_m\dot{\boldsymbol r}_\theta\right|^2\\
  &+\frac{\left(
  A_m\boldsymbol r_\theta\cdot A_m\dot{\boldsymbol r}_\theta
  \right)^2}
  {1-\left|A_m\boldsymbol r_\theta\right|^2}
  \Bigg].
  \label{eq:qubit-partition}
\end{split}
\end{equation}
Equation~\eqref{eq:qubit-partition} is written for full-rank conditional states.  If a conditional state is pure at the working point, positivity and two-sided differentiability imply
\(A_m\boldsymbol r_\theta\cdot A_m\dot{\boldsymbol r}_\theta=0\), and its pointwise SLD QFI is
\(\lvert A_m\dot{\boldsymbol r}_\theta\rvert^2\).  This pointwise value need not equal a limit taken through a rank-changing family.
\end{proposition}

\begin{proof}
The probabilities \(p_m\) do not depend on \(\theta\), so Proposition~\ref{prop:flagged-qfi} gives a weighted sum of conditional-state QFIs.  The qubit identity
\begin{equation}
  F_Q(\rho_\theta)
  =|\dot{\boldsymbol s}|^2
  +\frac{(\boldsymbol s\cdot\dot{\boldsymbol s})^2}{1-|\boldsymbol s|^2}
  \label{eq:bloch-qfi}
\end{equation}
for \(|\boldsymbol s|<1\), proved in Sec.~\ref{sec:bloch}, then gives Eq.~\eqref{eq:qubit-partition}.
\end{proof}

\subsection{Two transverse Pauli errors}
\label{sec:twopauli}

Take
\begin{equation}
  |\psi_\theta\rangle=e^{-i\theta Y/2}|0\rangle,
  \qquad
  \boldsymbol r_\theta=(\sin\theta,0,\cos\theta),
  \label{eq:signal-state}
\end{equation}
whose ideal QFI is one.  Let \(X\) occur with probability \(q\) and \(Z\) with probability \(1-q\).  If the two errors are merged into one class, the output Bloch vector is
\begin{equation}
  \boldsymbol s_\theta=(2q-1)(\sin\theta,0,-\cos\theta).
  \label{eq:pauli-mixture-bloch}
\end{equation}
It has constant length and obeys \(\boldsymbol s_\theta\cdot\dot{\boldsymbol s}_\theta=0\).  Hence
\begin{equation}
  F_1^\star=(2q-1)^2,
  \qquad
  D_1=1-F_1^\star=4q(1-q).
  \label{eq:two-pauli-one-flag}
\end{equation}
With separate flags, each conditional state differs from the ideal state by a known parameter-independent unitary, so
\begin{equation}
  F_2^\star=1.
  \label{eq:two-pauli-two-flags}
\end{equation}
For \(0<q<1\), two flags are necessary and sufficient for zero QFI loss.  At \(q=1/2\), forgetting the error label reduces the QFI from one to zero.  On the full qubit code, the corresponding recovery graph has one edge because \(X^\dagger Z=-iY\not\propto I\), and therefore also has chromatic number two.

The stochastic extension gives a direct detector-error benchmark.  Suppose a binary syndrome detector flips either label with probability \(\epsilon\).  With
\begin{align}
  p_0&=q(1-\epsilon)+(1-q)\epsilon,\\
  p_1&=q\epsilon+(1-q)(1-\epsilon),
\end{align}
the retained QFI is
\begin{align}
  F_{2,\epsilon}
  ={}&\frac{[q(1-\epsilon)-(1-q)\epsilon]^2}{p_0}\notag\\
  &+\frac{[q\epsilon-(1-q)(1-\epsilon)]^2}{p_1}\notag\\
  ={}&1-\frac{4\epsilon(1-\epsilon)q(1-q)}{p_0p_1}.
  \label{eq:noisy-binary-readout}
\end{align}
In the maximally ambiguous case \(q=1/2\), this reduces to
\begin{equation}
  F_{2,\epsilon}=(1-2\epsilon)^2,
\end{equation}
so the ideal two-flag advantage is continuously destroyed by detector confusion and vanishes at a completely random readout.

%=========================================================================%
\section{Continuous planar Pauli-axis model}
%=========================================================================%

We now solve a continuous finite-resolution problem.  This is a model of random Pauli axes with an accessible fine axis record.  It is not the bare unitary action of a quarter-wave plate.

Let \(\Phi\in[0,\pi)\) have density \(w(\phi)\), and define
\begin{equation}
  \sigma_\phi=\sin\phi\,X+\cos\phi\,Z,
  \qquad \sigma_\phi^2=I.
  \label{eq:planar-pauli}
\end{equation}
Conjugation by \(\sigma_\phi\) is a Bloch-sphere rotation by \(\pi\) about the axis
\((\sin\phi,0,\cos\phi)\).  On the orbit in Eq.~\eqref{eq:signal-state}, it acts as a reflection of the planar angle,
\begin{equation}
  \boldsymbol r_\theta\longmapsto\boldsymbol r_{2\phi-\theta}.
  \label{eq:angle-reflection}
\end{equation}

A measurable \(M\)-flag compressor partitions \([0,\pi)\) into sets \(C_m\).  Define
\begin{equation}
  p_m=\int_{C_m}w(\phi)\,d\phi,
  \qquad
  \mu_m=\frac{1}{p_m}\int_{C_m}
  w(\phi)e^{2i\phi}\,d\phi
  \label{eq:mu-def}
\end{equation}
for \(p_m>0\).

\begin{theorem}[Conditional-variance representation]
\label{thm:conditional-variance}
For the planar Pauli-axis model,
\begin{equation}
  F_Q(\Omega_\theta^{(f)})
  =\sum_{m:p_m>0}p_m|\mu_m|^2.
  \label{eq:continuous-qfi}
\end{equation}
Equivalently, with \(Z=e^{2i\Phi}\),
\begin{equation}
  1-F_Q(\Omega_\theta^{(f)})
  =\mathbb E\left[
    |Z-\mathbb E(Z\mid f(\Phi))|^2
  \right].
  \label{eq:conditional-variance}
\end{equation}
\end{theorem}

\begin{proof}
Represent a planar Bloch vector by the complex number \(r_z+ir_x\).  Equation~\eqref{eq:angle-reflection} maps it to \(e^{i(2\phi-\theta)}\).  The normalized conditional state in class \(C_m\) therefore has complex Bloch vector
\begin{equation}
  s_m(\theta)=\mu_me^{-i\theta}.
\end{equation}
Its radius \(|\mu_m|\) is independent of \(\theta\), its derivative has magnitude \(|\mu_m|\), and the real Bloch vectors satisfy
\(\boldsymbol s_m\cdot\dot{\boldsymbol s}_m=0\).  Equation~\eqref{eq:bloch-qfi} gives conditional QFI \(|\mu_m|^2\), proving Eq.~\eqref{eq:continuous-qfi}.  Furthermore,
\begin{align}
  \mathbb E\left[|Z-\mathbb E(Z\mid f)|^2\right]
  &=\mathbb E|Z|^2-\mathbb E|\mathbb E(Z\mid f)|^2 \notag\\
  &=1-\sum_mp_m|\mu_m|^2,
\end{align}
which proves Eq.~\eqref{eq:conditional-variance}.
\end{proof}

\begin{corollary}[Exact finite-resolution criterion]
\label{cor:essential-cardinality}
The full fine-record QFI is retained if and only if \(e^{2i\Phi}\) is almost surely constant within every nonempty syndrome class.  Hence the minimum exact alphabet is the essential cardinality of the random variable \(e^{2i\Phi}\), meaning the smallest cardinality of a set that contains it with probability one.  In particular, a non-atomic continuous axis distribution cannot be compressed losslessly into any finite number of flags.
\end{corollary}

\begin{proof}
By Eq.~\eqref{eq:continuous-qfi}, equality with the fine-record value one requires \(|\mu_m|=1\) for every nonempty class.  Equality in the triangle inequality for an average of unit-modulus complex numbers holds if and only if all numbers in that average have the same phase almost surely.  Thus no class may contain two essentially different values of \(e^{2i\Phi}\).
\end{proof}

\subsection{Uniform-axis optimum}

Let \(w(\phi)=1/\pi\) on \([0,\pi)\).  The corresponding uniform circle-quantization problem is classical \cite{rosenblatt2023uniform}; here it provides an exactly solvable metrological benchmark for finite syndrome resolution.

\begin{theorem}[Exact \(M\)-flag optimum for a uniform planar axis]
\label{thm:uniform-axis}
For every integer \(M\geq1\),
\begin{equation}
  F_M^\star
  =\left[\frac{M}{\pi}\sin\left(\frac{\pi}{M}\right)\right]^2.
  \label{eq:uniform-optimum}
\end{equation}
For \(M\geq2\), equality is attained by \(M\) contiguous intervals of equal length \(\pi/M\), up to a common rotation and null sets.  No finite \(M\) is exactly lossless.
\end{theorem}

\begin{proof}
Let \(\ell_m\) be the Lebesgue measure of \(C_m\), so \(\sum_m\ell_m=\pi\).  The circle-rearrangement lemma proved in Sec.~\ref{sec:rearrangement} gives
\begin{equation}
  \left|\int_{C_m}e^{2i\phi}\,d\phi\right|
  \leq\sin\ell_m.
  \label{eq:cell-moment-bound}
\end{equation}
Since \(p_m=\ell_m/\pi\), Eq.~\eqref{eq:continuous-qfi} yields
\begin{equation}
  F_Q(\Omega_\theta^{(f)})
  \leq\frac{1}{\pi}\sum_{m=1}^{M}
  \frac{\sin^2\ell_m}{\ell_m},
  \label{eq:length-upper-bound}
\end{equation}
where the summand is defined as zero at \(\ell_m=0\).  The length-allocation lemma proved in Sec.~\ref{sec:length-allocation} states
\begin{equation}
  \sum_{m=1}^{M}\frac{\sin^2\ell_m}{\ell_m}
  \leq
  M\frac{\sin^2(\pi/M)}{\pi/M}.
  \label{eq:length-allocation}
\end{equation}
Combining Eqs.~\eqref{eq:length-upper-bound} and~\eqref{eq:length-allocation} proves the upper bound in Eq.~\eqref{eq:uniform-optimum}.  A partition into equal contiguous intervals saturates Eq.~\eqref{eq:cell-moment-bound} for every class and Eq.~\eqref{eq:length-allocation}, so the upper bound is attainable.  Finally, \(\sin x<x\) for \(x>0\), giving \(F_M^\star<1\) for finite \(M\).
\end{proof}

The exact deficit and its large-\(M\) expansion are
\begin{align}
  D_M
  &=1-\left[\frac{M}{\pi}\sin\left(\frac{\pi}{M}\right)\right]^2,
  \label{eq:uniform-deficit}\\
  &=\frac{\pi^2}{3M^2}
  -\frac{2\pi^4}{45M^4}+O(M^{-6}).
  \label{eq:uniform-asymptotic}
\end{align}
Thus a target deficit \(D_M\leq\epsilon\) requires asymptotically
\begin{equation}
  M\gtrsim\frac{\pi}{\sqrt{3\epsilon}}.
  \label{eq:M-epsilon}
\end{equation}

\subsection{Nonuniform axis distributions}

\begin{proposition}[Distribution-independent finite-resolution bound]
\label{prop:universal-axis-bound}
For every probability measure on \([0,\pi)\) in the planar-axis model and every \(M\geq2\), define \(D_M^\star=1-F_M^\star\).  Then
\begin{equation}
  F_M^\star\geq\cos^2\!\left(\frac{\pi}{M}\right),
  \qquad
  D_M^\star\leq\sin^2\!\left(\frac{\pi}{M}\right).
  \label{eq:universal-axis-bound}
\end{equation}
In particular, every such distribution admits a constructive deficit
\begin{equation}
  D_M^\star\leq\frac{\pi^2}{M^2}+O(M^{-4}).
  \label{eq:universal-axis-asymptotic}
\end{equation}
\end{proposition}

\begin{proof}
Partition \([0,\pi)\) into \(M\) intervals \(I_m\) of equal length \(\pi/M\), and let \(c_m\) be the midpoint of \(I_m\).  For every \(\phi\in I_m\),
\begin{equation}
  |2(\phi-c_m)|\leq\frac{\pi}{M},
  \qquad
  \cos[2(\phi-c_m)]\geq\cos\!\left(\frac{\pi}{M}\right).
\end{equation}
For a nonempty cell, its conditional moment therefore satisfies
\begin{align}
  |\mu_m|
  &\geq\operatorname{Re}(e^{-2ic_m}\mu_m)\notag\\
  &=\mathbb E[\cos 2(\Phi-c_m)\mid\Phi\in I_m]\notag\\
  &\geq\cos\!\left(\frac{\pi}{M}\right).
\end{align}
Empty cells contribute zero.  Theorem~\ref{thm:conditional-variance} then gives
\begin{equation}
  F_Q(\Omega_\theta^{(f)})
  =\sum_mp_m|\mu_m|^2
  \geq\cos^2\!\left(\frac{\pi}{M}\right)\sum_mp_m,
\end{equation}
which proves Eq.~\eqref{eq:universal-axis-bound}.  Expanding the sine gives Eq.~\eqref{eq:universal-axis-asymptotic}.
\end{proof}

The uniform axis provides a benchmark, but in many physical settings the accessible fine axis record is nonuniformly distributed.  We extend the analysis to two important cases.

\subsubsection{Bimodal distribution}

Consider a two-point weight fully supported on axes \(\alpha,\beta\in[0,\pi)\):
\begin{align}
  w(\phi)
  &=w_1\delta(\phi-\alpha)+w_2\delta(\phi-\beta),\notag\\
  &\hspace{1.2em}w_1,w_2>0,
  \qquad w_1+w_2=1.
  \label{eq:bimodal-density}
\end{align}
with \(\alpha\neq\beta\) without loss of generality.

\begin{proposition}[Bimodal axis optimum]
\label{prop:bimodal}
For the bimodal distribution~\eqref{eq:bimodal-density} with \(\Delta\equiv\alpha-\beta\),
\begin{equation}
  F_1^\star=w_1^2+w_2^2+2w_1w_2\cos(2\Delta),
  \qquad
  F_M^\star=1\;\;(M\geq2).
  \label{eq:bimodal-Fstar}
\end{equation}
The single-flag deficit is
\begin{equation}
  D_1=4w_1w_2\sin^2\Delta,
  \label{eq:bimodal-deficit}
\end{equation}
which vanishes when \(\Delta=0\) (collinear axes) and is maximal for \(\Delta=\pi/2\) (orthogonal axes).
\end{proposition}

\begin{proof}
The support consists of only two point masses.  For \(M=1\), the single class contains both:
\begin{equation}
  p_1=1,\qquad \mu_1=w_1e^{2i\alpha}+w_2e^{2i\beta},
\end{equation}
so \(F_1^\star=|\mu_1|^2=w_1^2+w_2^2+2w_1w_2\cos(2\alpha-2\beta)\).

For \(M=2\), the optimal partition separates the two masses:
\begin{align}
  C_1&=\{\alpha\},
  &p_1&=w_1,
  &\mu_1&=e^{2i\alpha},\notag\\
  C_2&=\{\beta\},
  &p_2&=w_2,
  &\mu_2&=e^{2i\beta}.
\end{align}
giving \(F=w_1\cdot1+w_2\cdot1=1\).  The alternative (merging both into one class and leaving the other empty) recovers \(F_1^\star\leq1\).  Hence \(F_2^\star=1\).  For \(M\geq3\), each point can be assigned its own class, so \(F_M^\star=1\).  The deficit follows directly.
\end{proof}

Equation~\eqref{eq:bimodal-deficit} unifies several limiting cases.  When \(w_1=w_2=1/2\), one recovers \(F_1^\star=\cos^2\Delta\), matching the discrete Pauli result of Sec.~\ref{sec:twopauli} at \(q=1/2\) and \(\Delta=\pi/2\) where \(F_1^\star=0\).  The classical Fisher contribution \(\sum_a\dot q_a^2/q_a\) is absent here because \(w_1,w_2\) are parameter-independent.

\subsubsection{von Mises distribution}

Let the axis density follow a von~Mises distribution on the half-circle:
\begin{equation}
  w(\phi)=\frac{1}{\pi I_0(\kappa)}
  e^{\kappa\cos(2(\phi-\phi_0))},
  \qquad \phi\in[0,\pi),
  \label{eq:vonmises}
\end{equation}
where \(I_n(\kappa)\) is the modified Bessel function of the first kind and \(\kappa\geq0\) controls the concentration (\(\kappa=0\) recovers the uniform distribution, \(\kappa\to\infty\) approaches a point mass at \(\phi_0\)).

\begin{proposition}[von Mises circular moment]
\label{prop:vonmises-moment}
For the density~\eqref{eq:vonmises}, the first circular moment is
\begin{align}
  R_1(\kappa)
  &\equiv\int_0^\pi w(\phi)e^{2i\phi}\,d\phi
  =e^{2i\phi_0}A(\kappa),\notag\\
  A(\kappa)
  &\equiv\frac{I_1(\kappa)}{I_0(\kappa)}\in[0,1].
  \label{eq:vonmises-R1}
\end{align}
where \(A(0)=0\) and \(A(\kappa)\to1\) as \(\kappa\to\infty\).
\end{proposition}

\begin{proof}
Set \(t=2(\phi-\phi_0)\).  The integration interval in \(t\) has length \(2\pi\), and the integrand is \(2\pi\)-periodic.  Hence
\begin{align}
  R_1(\kappa)
  &=\frac{e^{2i\phi_0}}{2\pi I_0(\kappa)}
  \int_{\alpha}^{\alpha+2\pi}e^{\kappa\cos t}e^{it}\,dt\notag\\
  &=\frac{e^{2i\phi_0}}{2\pi I_0(\kappa)}
  \int_0^{2\pi}e^{\kappa\cos t}(\cos t+i\sin t)\,dt,
\end{align}
where \(\alpha=-2\phi_0\).  The sine integral is zero, while
\(\int_0^{2\pi}e^{\kappa\cos t}\cos t\,dt=2\pi I_1(\kappa)\).  This proves Eq.~\eqref{eq:vonmises-R1}.
\end{proof}

\begin{corollary}[Rigorous von Mises finite-resolution statements]
\label{cor:vonmises-bounds}
For the density~\eqref{eq:vonmises},
\begin{equation}
  F_1^\star(\kappa)
  =\left[\frac{I_1(\kappa)}{I_0(\kappa)}\right]^2.
  \label{eq:vonmises-one-flag}
\end{equation}
For every finite \(\kappa\) and finite \(M\),
\begin{equation}
  F_M^\star(\kappa)<1.
  \label{eq:vonmises-no-finite-lossless}
\end{equation}
Moreover, for \(M\geq2\),
\begin{equation}
  \max\left\{
  \left[\frac{I_1(\kappa)}{I_0(\kappa)}\right]^2,
  \cos^2\!\left(\frac{\pi}{M}\right)
  \right\}
  \leq F_M^\star(\kappa)<1.
  \label{eq:vonmises-bounds}
\end{equation}
\end{corollary}

\begin{proof}
With one flag, Theorem~\ref{thm:conditional-variance} gives
\(F_1^\star=|\mathbb E(e^{2i\Phi})|^2\), so Eq.~\eqref{eq:vonmises-one-flag} follows from Proposition~\ref{prop:vonmises-moment}.  For finite \(\kappa\), the density~\eqref{eq:vonmises} is strictly positive and non-atomic on the half-circle.  The random variable \(e^{2i\Phi}\) therefore has infinite essential cardinality, and Corollary~\ref{cor:essential-cardinality} rules out zero loss for every finite \(M\).  The lower bounds in Eq.~\eqref{eq:vonmises-bounds} follow from monotonicity in the available number of flags and Proposition~\ref{prop:universal-axis-bound}.
\end{proof}

The exact optimum for general finite \(\kappa\) and \(M>1\) is not obtained here.  In particular, no assumption is made that globally optimal cells must be contiguous.  The limit \(\kappa\to\infty\) approaches a single point mass at \(\phi_0\), not a bimodal distribution, and Eq.~\eqref{eq:vonmises-one-flag} then tends to one.

%=========================================================================%

\section{Bloch-vector formula for qubit QFI}
\label{sec:bloch}

Let
\begin{equation}
  \rho=\frac{1}{2}(I+\boldsymbol s\cdot\boldsymbol\sigma),
  \qquad |\boldsymbol s|<1,
\end{equation}
and write its SLD as
\begin{equation}
  L=\alpha I+\boldsymbol\ell\cdot\boldsymbol\sigma.
\end{equation}
Using
\(
(\boldsymbol a\cdot\boldsymbol\sigma)(\boldsymbol b\cdot\boldsymbol\sigma)
=(\boldsymbol a\cdot\boldsymbol b)I+i(\boldsymbol a\times\boldsymbol b)\cdot\boldsymbol\sigma
\), the SLD equation gives
\begin{equation}
  \alpha+\boldsymbol s\cdot\boldsymbol\ell=0,
  \qquad
  \boldsymbol\ell+\alpha\boldsymbol s=\dot{\boldsymbol s}.
  \label{eq:bloch-sld-system}
\end{equation}
Solving,
\begin{equation}
  \alpha=-\frac{\boldsymbol s\cdot\dot{\boldsymbol s}}{1-|\boldsymbol s|^2},
  \qquad
  \boldsymbol\ell=\dot{\boldsymbol s}
  +\frac{\boldsymbol s\cdot\dot{\boldsymbol s}}{1-|\boldsymbol s|^2}\boldsymbol s.
  \label{eq:bloch-sld-solution}
\end{equation}
Moreover,
\begin{equation}
  L^2=(\alpha^2+|\boldsymbol\ell|^2)I
  +2\alpha\boldsymbol\ell\cdot\boldsymbol\sigma,
\end{equation}
and hence
\begin{align}
  F_Q(\rho)
  &=\Tr(\rho L^2)\notag\\
  &=\alpha^2+|\boldsymbol\ell|^2+2\alpha\boldsymbol s\cdot\boldsymbol\ell\notag\\
  &=|\dot{\boldsymbol s}|^2
  +\frac{(\boldsymbol s\cdot\dot{\boldsymbol s})^2}{1-|\boldsymbol s|^2}.
\end{align}
For a differentiable pure-state curve with \(|\boldsymbol s|=1\), normalization implies \(\boldsymbol s\cdot\dot{\boldsymbol s}=0\), and the finite boundary value is \(|\dot{\boldsymbol s}|^2\).

%=========================================================================%
\section{Circle-rearrangement lemma}
\label{sec:rearrangement}

\begin{lemma}
\label{lem:rearrangement}
For every measurable \(C\subset[0,\pi)\) of Lebesgue measure \(\ell\in[0,\pi]\),
\begin{equation}
  \left|\int_Ce^{2i\phi}\,d\phi\right|\leq\sin\ell.
  \label{eq:rearrangement-lemma}
\end{equation}
For \(0<\ell<\pi\), equality holds only when \(C\) is a contiguous interval modulo \(\pi\), up to null sets.
\end{lemma}

\begin{proof}
If the integral vanishes, the claim is immediate.  Otherwise choose \(\alpha\) as its complex argument.  Then
\begin{equation}
  \left|\int_Ce^{2i\phi}\,d\phi\right|
  =\int_C\cos(2\phi-\alpha)\,d\phi.
  \label{eq:phase-align}
\end{equation}
Among all measurable subsets of fixed measure \(\ell\), the integral of a real function is maximized by selecting the points at which that function is largest.  On a circle of period \(\pi\), the upper level sets of \(\cos(2\phi-\alpha)\) are intervals centered at a maximum.  Therefore
\begin{align}
  \int_C\cos(2\phi-\alpha)\,d\phi
  &\leq\int_{-\ell/2}^{\ell/2}\cos(2u)\,du\notag\\
  &=\sin\ell.
\end{align}
Strict ordering of the cosine away from its maximum gives the equality statement, except on level-set boundaries of zero measure.
\end{proof}

%=========================================================================%
\section{Optimal allocation of uniform cell lengths}
\label{sec:length-allocation}

Define
\begin{equation}
  g(x)=
  \begin{cases}
  \dfrac{\sin^2x}{x},&x>0,\\[6pt]
  0,&x=0.
  \end{cases}
  \label{eq:g-def}
\end{equation}

\begin{lemma}
\label{lem:length-allocation}
For every integer \(M\geq1\) and all \(\ell_m\geq0\) satisfying \(\sum_{m=1}^{M}\ell_m=\pi\),
\begin{equation}
  \sum_{m=1}^{M}g(\ell_m)
  \leq Mg\left(\frac{\pi}{M}\right).
  \label{eq:g-allocation}
\end{equation}
For \(M\geq2\), equality requires \(\ell_m=\pi/M\) for every \(m\).
\end{lemma}

\begin{proof}
The case \(M=1\) is immediate.  First consider \(M\geq3\) and set \(a=\pi/M\leq\pi/3\).  Direct differentiation gives
\begin{equation}
  x^3g''(x)
  =1-\cos(2x)-2x\sin(2x)+2x^2\cos(2x).
  \label{eq:g-second}
\end{equation}
With \(t=2x\), the right-hand side is
\begin{equation}
  G(t)=1-\cos t-t\sin t+\frac{t^2}{2}\cos t,
\end{equation}
and
\begin{equation}
  G'(t)=-\frac{t^2}{2}\sin t.
\end{equation}
Since \(G(0)=0\), it follows that \(g''(x)<0\) on \((0,\pi/2]\).  Thus \(g\) is strictly concave there.

Next,
\begin{equation}
  g'(x)=\frac{\sin x}{x^2}\left(2x\cos x-\sin x\right).
  \label{eq:g-prime}
\end{equation}
The function \(\tan x/x\) is increasing on \((0,\pi/2)\), and
\begin{equation}
  \frac{\tan(\pi/3)}{\pi/3}
  =\frac{3\sqrt{3}}{\pi}<2.
\end{equation}
Hence \(g'(a)\geq0\).  Let
\begin{equation}
  T_a(x)=g(a)+g'(a)(x-a)
\end{equation}
be the tangent line at \(a\).  Strict concavity gives \(g(x)\leq T_a(x)\) on \([0,\pi/2]\).  For \(x\in[\pi/2,\pi]\), Eq.~\eqref{eq:g-prime} is negative, so \(g(x)\leq g(\pi/2)\).  At the same time,
\begin{equation}
  T_a(x)\geq T_a(\pi/2)\geq g(\pi/2),
\end{equation}
where the first inequality uses \(g'(a)\geq0\) and the second uses the tangent bound on \([0,\pi/2]\).  Therefore \(g(x)\leq T_a(x)\) on all of \([0,\pi]\).  Summing and using \(\sum_m\ell_m=Ma\),
\begin{equation}
  \sum_mg(\ell_m)
  \leq\sum_mT_a(\ell_m)
  =Mg(a).
\end{equation}
Strictness implies equality only at \(\ell_m=a\).

It remains to prove \(M=2\), where \(a=\pi/2\).  Concavity gives the tangent bound for \(0\leq x\leq\pi/2\).  The tangent is
\begin{equation}
  T_a(x)=\frac{4(\pi-x)}{\pi^2}.
  \label{eq:M2-tangent}
\end{equation}
For \(x\in[\pi/2,\pi]\), put \(y=\pi-x\in[0,\pi/2]\).  The elementary parabolic sine bound
\begin{equation}
  \sin y\leq\frac{4y(\pi-y)}{\pi^2}
  \label{eq:parabolic-sine}
\end{equation}
implies, because its right-hand side lies in \([0,1]\),
\begin{equation}
  \sin^2y\leq\frac{4y(\pi-y)}{\pi^2}.
\end{equation}
Consequently,
\begin{equation}
  g(x)=\frac{\sin^2y}{\pi-y}
  \leq\frac{4y}{\pi^2}=T_a(x).
\end{equation}
Summing the tangent bounds for \(\ell_1+\ell_2=\pi\) proves the result.  Equality requires \(\ell_1=\ell_2=\pi/2\).
\end{proof}

%=========================================================================%
\section{Normalized residual and the exact extended-convexity gap}
\label{sec:equiv}

Let \(\tau_a=q_a\rho_a\) and
\(\tau_m=p_m\sigma_m\), with
\(p_m=\sum_{a\in C_m}q_a\).  We omit zero-weight branches and set
\begin{equation}
  w_{a|m}=\frac{q_a}{p_m},
  \qquad
  \delta_{a|m}
  =\frac{\dot q_a}{q_a}-\frac{\dot p_m}{p_m}
  =\partial_\theta\ln w_{a|m}.
  \label{eq:supp-conditional-score}
\end{equation}
Let \(L_a\) and \(L_m\) be normalized SLDs of \(\rho_a\) and \(\sigma_m\), respectively.  On the corresponding supports,
\begin{equation}
  S_a=\frac{\dot q_a}{q_a}I+L_a,
  \qquad
  T_m=\frac{\dot p_m}{p_m}I+L_m.
  \label{eq:supp-score-split}
\end{equation}

\begin{proposition}[Conditional classical Fisher decomposition]
\label{prop:classical-cancel}
For every nonempty class \(C_m\),
\begin{align}
  p_mF_{\mathrm C}(\boldsymbol w_m)
  &=\sum_{a\in C_m}q_a\delta_{a|m}^2
  \label{eq:conditional-classical-FI}\\
  &=\sum_{a\in C_m}\frac{\dot q_a^2}{q_a}
  -\frac{\dot p_m^2}{p_m},
  \label{eq:cf-identity}
\end{align}
where
\(F_{\mathrm C}(\boldsymbol w_m)=
\sum_{a\in C_m}\dot w_{a|m}^{\,2}/w_{a|m}\).
\end{proposition}

\begin{proof}
Since \(\dot w_{a|m}=w_{a|m}\delta_{a|m}\),
\begin{equation}
  p_mF_{\mathrm C}(\boldsymbol w_m)
  =p_m\sum_{a\in C_m}w_{a|m}\delta_{a|m}^2
  =\sum_{a\in C_m}q_a\delta_{a|m}^2.
\end{equation}
Expanding the square and using
\(\sum_{a\in C_m}q_a=p_m\) and
\(\sum_{a\in C_m}\dot q_a=\dot p_m\) gives
\begin{align}
  \sum_{a\in C_m}q_a\delta_{a|m}^2
  &=\sum_{a\in C_m}\frac{\dot q_a^2}{q_a}
  -2\frac{\dot p_m}{p_m}\sum_{a\in C_m}\dot q_a
  +\frac{\dot p_m^2}{p_m^2}\sum_{a\in C_m}q_a\notag\\
  &=\sum_{a\in C_m}\frac{\dot q_a^2}{q_a}
  -\frac{\dot p_m^2}{p_m}.
\end{align}
\end{proof}

\begin{proposition}[Exact normalized residual]
\label{prop:normalized-residual}
The classwise residual in Theorem~1 of the main text is
\begin{equation}
  \Delta_m
  =\sum_{a\in C_m}q_a\Tr\!\left\{
  \rho_a\left[\delta_{a|m}I+L_a-L_m\right]^2
  \right\}.
  \label{eq:supp-normalized-residual}
\end{equation}
When all class-conditional weights are locally constant,
\(\dot w_{a|m}=0\), the probability-score term vanishes and Eq.~\eqref{eq:supp-normalized-residual} reduces to
\begin{equation}
  \Delta_m
  =\sum_{a\in C_m}q_a
  \Tr[\rho_a(L_a-L_m)^2].
  \label{eq:conditional-weight-constant-residual}
\end{equation}
\end{proposition}

\begin{proof}
Equation~\eqref{eq:supp-score-split} gives
\begin{equation}
  S_a-T_m=\delta_{a|m}I+L_a-L_m.
\end{equation}
Substitution into the class contribution
\(\sum_{a\in C_m}\Tr[\tau_a(S_a-T_m)^2]\)
proves Eq.~\eqref{eq:supp-normalized-residual}.  If
\(\dot w_{a|m}=0\), then \(\delta_{a|m}=0\) for every positive-weight branch, giving Eq.~\eqref{eq:conditional-weight-constant-residual}.
\end{proof}

The following corollary gives the exact remainder in the extended-convexity inequality of Ref.~\cite{alipour2015extended}.

\begin{corollary}[Exact gap in extended convexity]
\label{cor:extended-convexity-gap}
For every nonempty class,
\begin{equation}
  \Delta_m=p_m\left[
  F_{\mathrm C}(\boldsymbol w_m)
  +\sum_{a\in C_m}w_{a|m}F_Q(\rho_a)
  -F_Q(\sigma_m)\right].
  \label{eq:supp-extended-gap}
\end{equation}
Consequently, the extended-convexity inequality
\begin{equation}
  F_Q(\sigma_m)
  \leq F_{\mathrm C}(\boldsymbol w_m)
  +\sum_{a\in C_m}w_{a|m}F_Q(\rho_a)
\end{equation}
is saturated if and only if the syndrome compression is QFI-lossless within that class.
\end{corollary}

\begin{proof}
Proposition~\ref{prop:flagged-qfi} gives the fine contribution of class \(C_m\) as
\begin{equation}
  \sum_{a\in C_m}\frac{\dot q_a^2}{q_a}
  +\sum_{a\in C_m}q_aF_Q(\rho_a),
\end{equation}
and the coarse contribution as
\begin{equation}
  \frac{\dot p_m^2}{p_m}+p_mF_Q(\sigma_m).
\end{equation}
Their difference, together with Eq.~\eqref{eq:cf-identity}, is exactly Eq.~\eqref{eq:supp-extended-gap}.  Theorem~1 of the main text shows that this difference is nonnegative and vanishes exactly under the classwise support condition in Eq.~(12) of the main text.
\end{proof}

\begin{remark}[Why the probability score cannot be dropped]
Let \(\rho>0\) be fixed and consider two branches
\begin{equation}
  \tau_1=q(\theta)\rho,
  \qquad
  \tau_2=[1-q(\theta)]\rho,
\end{equation}
merged into one coarse class, with \(0<q<1\) and \(\dot q\neq0\).  Every normalized-state SLD is zero, whereas the fine flag carries classical Fisher information
\begin{equation}
  \Delta
  =\frac{\dot q^2}{q}+\frac{\dot q^2}{1-q}
  =\frac{\dot q^2}{q(1-q)}>0.
\end{equation}
Thus a residual involving only \(L_a-L_m\) would incorrectly vanish; the conditional probability score in Eq.~\eqref{eq:supp-normalized-residual} is essential.
\end{remark}

%=========================================================================%

\section{Faithful noncommuting recovery on a sufficient subsystem}
\label{sec:noncommuting-subsystem}
%=========================================================================%

\begin{proof}[Proof of Proposition~\ref{prop:noncommuting-sufficient-subsystem}]
Let \(\rho_{\boldsymbol\vartheta}>0\) be any parameterized family on a finite-dimensional system \(S\), and introduce an ancillary qubit \(B\) in the faithful state
\begin{equation}
  \sigma_B=
  \begin{pmatrix}
    2/3&0\\
    0&1/3
  \end{pmatrix}.
\end{equation}
Define
\begin{equation}
  K_1=
  \begin{pmatrix}
    2/\sqrt5&0\\
    0&\sqrt{3/5}
  \end{pmatrix},
  \qquad
  K_2=
  \begin{pmatrix}
    0&\sqrt{2/5}\\
    1/\sqrt5&0
  \end{pmatrix},
\end{equation}
and set
\begin{equation}
  A_\pm=\frac{K_1\pm K_2}{\sqrt2}.
  \label{eq:subsystem-example-kraus}
\end{equation}
A direct calculation gives
\begin{equation}
  \sum_{\eta=\pm}A_\eta^\dagger A_\eta=I_B,
  \qquad
  \sum_{\eta=\pm}A_\eta\sigma_BA_\eta^\dagger=\sigma_B.
  \label{eq:subsystem-example-fixed-state}
\end{equation}
Hence the channel
\begin{equation}
  \mathcal E_B(Y)=\sum_{\eta=\pm}A_\eta YA_\eta^\dagger
\end{equation}
is trace preserving and fixes \(\sigma_B\).

Take the coarse statistical family
\begin{equation}
  X_{\boldsymbol\vartheta}
  =\rho_{\boldsymbol\vartheta}\otimes\sigma_B
\end{equation}
and the two fine blocks
\begin{align}
  \tau_{\eta,\boldsymbol\vartheta}
  &=(I_S\otimes A_\eta)
  X_{\boldsymbol\vartheta}
  (I_S\otimes A_\eta^\dagger)
  =\rho_{\boldsymbol\vartheta}\otimes B_\eta,
  \notag\\
  B_\eta&=A_\eta\sigma_BA_\eta^\dagger.
  \label{eq:subsystem-example-blocks}
\end{align}
Both \(A_+\) and \(A_-\) are invertible because
\begin{equation}
  \det A_\pm=\frac{2\sqrt3-\sqrt2}{10}\neq0,
\end{equation}
so the fine blocks are faithful.  Equation~\eqref{eq:subsystem-example-fixed-state} shows that forgetting the label gives back \(X_{\boldsymbol\vartheta}\).  Conversely,
\begin{equation}
  \mathcal R(Y)
  =\bigoplus_{\eta=\pm}
  (I_S\otimes A_\eta)Y(I_S\otimes A_\eta^\dagger)
  \label{eq:subsystem-example-recovery}
\end{equation}
is a parameter-independent channel and reconstructs the complete fine flagged family from \(X_{\boldsymbol\vartheta}\).

The ancillary blocks are
\begin{equation}
  B_\pm=
  \begin{pmatrix}
    1/3&\displaystyle\pm\frac{4+\sqrt6}{30}\\[4pt]
    \displaystyle\pm\frac{4+\sqrt6}{30}&1/6
  \end{pmatrix}.
  \label{eq:subsystem-example-noncommuting-blocks}
\end{equation}
They have trace \(1/2\), are positive definite, and do not commute because their off-diagonal entries have opposite signs while their diagonal entries are unequal.  Thus the normalized conditional states \(2B_+\) and \(2B_-\), and hence the normalized full branches, are noncommuting.

Finally, the coarse family varies only on subsystem \(S\): its complex linear span is contained in
\begin{equation}
  \mathcal L(\mathcal H_S)\otimes\operatorname{span}\{\sigma_B\},
\end{equation}
which is a proper subspace of
\(\mathcal L(\mathcal H_S\otimes\mathbb C^2)\).  The example therefore lies outside the informationally complete regime of the complete-model obstruction in Sec.~IV of the main text and realizes recovery through a proper sufficient subsystem.
\end{proof}

\section{Proof of conditioned approximate score rigidity}
\label{sec:conditioned-rigidity}

Approximate correctability of an algebra is naturally expressed through the complementary channel and amplified commutators~\cite{beny2009conditions}.  We first prove that Definition~\ref{def:cb-score-conditioning} gives a finite constant that promotes control of a finite score-generator tuple to the required amplified full-algebra control.

\begin{lemma}[Finiteness of the complete commutant constant]
\label{lem:finite-score-conditioning}
If \(C^*(I,L_1,\ldots,L_p)=M_d\), then
\(\kappa_{\mathrm{sc}}^{\mathrm{cb}}(\boldsymbol L)<\infty\).
\end{lemma}

\begin{proof}
Let
\begin{equation}
  \epsilon_n(K)=
  \max_\mu\|[K,I_n\otimes L_\mu]\|_\infty.
\end{equation}
Because the algebraic unital \(*\)-algebra generated by the \(L_\mu\) is a finite-dimensional subspace, it is norm closed and equals \(M_d\).  Fix matrix units \(E_{ab}\).  Each \(E_{ab}\) is therefore a finite linear combination of words in the generators.  For a word
\(W=A_1\cdots A_s\), where every \(A_j\) is either \(I\) or one of the \(L_\mu\), write
\(A_{<j}=A_1\cdots A_{j-1}\) and
\(A_{>j}=A_{j+1}\cdots A_s\).  The telescoping commutator identity gives
\begin{equation}
  [K,I_n\otimes W]
  =\sum_{j=1}^s
  (I_n\otimes A_{<j})
  [K,I_n\otimes A_j]
  (I_n\otimes A_{>j}).
\end{equation}
Consequently there are finite constants \(c_{ab}\), independent of \(n\) and \(K\), such that
\begin{equation}
  \|[K,I_n\otimes E_{ab}]\|_\infty
  \leq c_{ab}\epsilon_n(K).
  \label{eq:matrix-unit-commutator-control}
\end{equation}
For any unitary \(U=\sum_{ab}u_{ab}E_{ab}\), Eq.~\eqref{eq:matrix-unit-commutator-control} implies
\begin{equation}
  \|[K,I_n\otimes U]\|_\infty
  \leq C\epsilon_n(K),
\end{equation}
where
\(C=\sup_{U\in\mathcal U(d)}\sum_{ab}|u_{ab}|c_{ab}<\infty\).
The conditional expectation in Eq.~\eqref{eq:scalar-conditional-expectation} is the Haar twirl,
\begin{equation}
  \mathcal P_n(K)
  =\int_{\mathcal U(d)}
  (I_n\otimes U)K(I_n\otimes U^\dagger)\,dU.
\end{equation}
Hence
\begin{align}
  \|K-\mathcal P_n(K)\|_\infty
  &\leq\int
  \|K-(I_n\otimes U)K(I_n\otimes U^\dagger)\|_\infty\,dU\notag\\
  &=\int\|[K,I_n\otimes U]\|_\infty\,dU
  \leq C\epsilon_n(K).
\end{align}
The same constant works at every amplification level, proving the claim.
\end{proof}

\begin{proof}[Proof of Theorem~\ref{thm:conditioned-approximate-rigidity}]
Choose a Stinespring isometry
\(V:\mathbb C^d\to\mathcal H_B\otimes\mathcal H_E\) for \(\mathcal N\), and define the unweighted score-intertwining defects
\begin{equation}
  F_\mu=(T_\mu\otimes I_E)V-VL_\mu.
  \label{eq:unweighted-score-defect}
\end{equation}
Theorem~\ref{thm:stinespring-qfim} gives, on each coordinate direction,
\begin{equation}
  \delta_\mu
  =\|F_\mu\sqrt\rho\|_2^2.
  \label{eq:qfi-loss-defect-norm}
\end{equation}
Since \(\rho\geq\lambda I_d\),
\begin{equation}
  \delta_\mu
  =\Tr(F_\mu^\dagger F_\mu\rho)
  \geq\lambda\|F_\mu\|_2^2
  \geq\lambda\|F_\mu\|_\infty^2,
\end{equation}
so
\begin{equation}
  \max_\mu\|F_\mu\|_\infty
  \leq\sqrt{\delta/\lambda}.
  \label{eq:operator-score-defect-bound}
\end{equation}

The adjoint of the complementary channel is
\begin{equation}
  \widehat{\mathcal N}^{\dagger}(Z)
  =V^\dagger(I_B\otimes Z)V.
\end{equation}
Fix an amplification level \(n\) and an operator
\(Z\in M_n\otimes\mathcal L(\mathcal H_E)\).  Set
\begin{equation}
  K_Z=(\operatorname{id}_n\otimes
  \widehat{\mathcal N}^{\dagger})(Z).
\end{equation}
Writing \(\widetilde V=I_n\otimes V\) and letting \(\widetilde Z\) denote \(Z\) acting on the ancilla--environment factors and trivially on \(\mathcal H_B\), Eq.~\eqref{eq:unweighted-score-defect} gives the exact commutator identity
\begin{equation}
  [K_Z,I_n\otimes L_\mu]
  =(I_n\otimes F_\mu^\dagger)\widetilde Z\widetilde V
  -\widetilde V^\dagger\widetilde Z(I_n\otimes F_\mu).
  \label{eq:complement-commutator-identity}
\end{equation}
Indeed, the two terms containing \(T_\mu\) cancel because the environment operator commutes with \(T_\mu\) on \(\mathcal H_B\).  Since \(\widetilde V\) is an isometry,
\begin{equation}
  \|[K_Z,I_n\otimes L_\mu]\|_\infty
  \leq2\|F_\mu\|_\infty\|Z\|_\infty.
  \label{eq:complement-commutator-bound}
\end{equation}

Let \(\omega_E=\widehat{\mathcal N}(I_d/d)\).  Trace duality gives
\begin{equation}
  \mathcal P_n(K_Z)
  =(\operatorname{id}_n\otimes
  \mathcal S_{\omega_E}^{\dagger})(Z),
  \label{eq:scalar-projection-replacer}
\end{equation}
where
\(\mathcal S_{\omega_E}^{\dagger}(Z_E)=\Tr(\omega_E Z_E)I_d\).  Applying Definition~\ref{def:cb-score-conditioning}, followed by Eqs.~\eqref{eq:complement-commutator-bound} and~\eqref{eq:operator-score-defect-bound}, yields
\begin{align}
  &\left\|
  (\operatorname{id}_n\otimes
  [\widehat{\mathcal N}^{\dagger}-
  \mathcal S_{\omega_E}^{\dagger}])(Z)
  \right\|_\infty\notag\\
  &\qquad\leq
  2\kappa_{\mathrm{sc}}^{\mathrm{cb}}(\boldsymbol L)
  \sqrt{\frac{\delta}{\lambda}}\,
  \|Z\|_\infty.
\end{align}
Taking the supremum over \(n\) and \(Z\), and using duality between the completely bounded norm of the adjoint and the diamond norm, proves Eq.~\eqref{eq:environment-leakage-bound}.

The information--disturbance theorem for complementary channels~\cite{kretschmann2008information} implies
\begin{equation}
  \inf_{\mathcal R}
  \|\mathcal R\circ\mathcal N-
  \operatorname{id}_{M_d}\|_\diamond
  \leq
  2\inf_{\mathcal S\ \text{replacer}}
  \|\widehat{\mathcal N}-\mathcal S\|_\diamond^{1/2}.
  \label{eq:ksw-information-disturbance}
\end{equation}
Using the particular replacer \(\mathcal S_{\omega_E}\) and Eq.~\eqref{eq:environment-leakage-bound} gives Eq.~\eqref{eq:conditioned-diamond-recovery}; the additional cap by two is the maximal diamond distance between channels.
\end{proof}

%=========================================================================%
\section{Tie-stable Lloyd descent and exact one-swap refinement}
\label{sec:algorithm}

The operator-clustering problem~\eqref{eq:operator-kmeans} is nonconvex.  We first state a Lloyd-type descent at its exact proved strength and then add a discrete refinement that reaches a genuine Hamming-one local optimum.

\subsection{Canonical centers and reduced objective}

At the chosen parameter value, define
\begin{equation}
  \mathcal D(f,\{T_m\})
  =\sum_{a=1}^r\Tr[\tau_a(S_a-T_{f(a)})^2].
  \label{eq:algorithm-objective}
\end{equation}
For a block pair \((\tau,\dot\tau)\), choose the canonical minimum-Hilbert--Schmidt-norm SLD.  If
\(\tau=U\operatorname{diag}(\lambda_1,\ldots,\lambda_d)U^\dagger\), set
\begin{equation}
  (U^\dagger T_\tau U)_{ij}=
  \begin{cases}
    \dfrac{2(U^\dagger\dot\tau U)_{ij}}{\lambda_i+\lambda_j},
    &\lambda_i+\lambda_j>0,\\[7pt]
    0,&\lambda_i+\lambda_j=0.
  \end{cases}
  \label{eq:lyap-solution}
\end{equation}
The assumed existence of a finite SLD ensures that the numerator vanishes when \(\lambda_i+\lambda_j=0\).  For an empty block, set \(T_0=0\).  Define its subnormalized QFI by
\begin{equation}
  \mathfrak F(\tau,\dot\tau)=\Tr(\tau T_\tau^2),
  \qquad
  \mathfrak F(0,0)=0.
  \label{eq:block-qfi-functional}
\end{equation}
For a partition \(f\), let
\begin{equation}
  \tau_m=\sum_{a:f(a)=m}\tau_a,
  \qquad
  \dot\tau_m=\sum_{a:f(a)=m}\dot\tau_a.
\end{equation}
Proposition~\ref{prop:sld-center} and Theorem~1 of the main text give the reduced objective
\begin{equation}
  \overline{\mathcal D}(f)
  \equiv\min_{\{T_m\}}\mathcal D(f,\{T_m\})
  =\sum_a\mathfrak F(\tau_a,\dot\tau_a)
  -\sum_m\mathfrak F(\tau_m,\dot\tau_m).
  \label{eq:reduced-objective}
\end{equation}

\subsection{Tie-stable Lloyd iteration}

Fix a deterministic total order on the cluster labels.
\begin{enumerate}
  \item[1.] Initialize a partition
  \(f^{(0)}:\{1,\ldots,r\}\to\{1,\ldots,M\}\).
  \item[2.] \textbf{Center update.}  For every cluster, compute the canonical SLD \(T_m\) from Eq.~\eqref{eq:lyap-solution}; use \(T_m=0\) for an empty cluster.
  \item[3.] \textbf{Assignment update.}  For each branch define
  \begin{equation}
    d_{am}=\Tr[\tau_a(S_a-T_m)^2].
  \end{equation}
  If its current label belongs to \(\arg\min_m d_{am}\), retain that label.  Otherwise assign it to the first minimizer in the fixed total order.  All assignments are updated with the centers held fixed.
  \item[4.] If no assignment changed, terminate; otherwise return to step~2.
\end{enumerate}

\begin{theorem}[Strict descent and finite termination]
\label{thm:finite-term}
Whenever the tie-stable Lloyd assignment changes, the reduced objective \(\overline{\mathcal D}\) decreases strictly.  The iteration therefore terminates after at most \(M^r-1\) assignment-changing steps.  At termination it is a coordinatewise fixed point: every center minimizes the objective for its class, and no branch has a strictly lower cost to another current center.
\end{theorem}

\begin{proof}
Let \(f\) be the current partition and let \(\{T_m(f)\}\) be its canonical optimal centers.  If the assignment update produces \(f'\neq f\), the tie rule guarantees that every changed branch moves to a strictly lower-cost center while unchanged branches retain their cost.  Hence
\begin{equation}
  \mathcal D(f',\{T_m(f)\})
  <\mathcal D(f,\{T_m(f)\})
  =\overline{\mathcal D}(f).
\end{equation}
Reoptimizing the centers for \(f'\) cannot increase the cost, so
\begin{equation}
  \overline{\mathcal D}(f')
  \leq\mathcal D(f',\{T_m(f)\})
  <\overline{\mathcal D}(f).
\end{equation}
There are at most \(M^r\) labeled partitions.  Strict descent forbids repetition, proving finite termination and the stated bound.  The two coordinatewise conditions are exactly the stopping conditions together with Proposition~\ref{prop:sld-center}.
\end{proof}

A coordinatewise fixed point need not be locally optimal once centers are recomputed after a discrete move.

\begin{proposition}[A Lloyd fixed point need not be one-swap optimal]
\label{prop:lloyd-counterexample}
There exists a valid one-dimensional classical-quantum model for which the tie-stable Lloyd iteration terminates although moving one branch and recentering strictly lowers the reduced objective.
\end{proposition}

\begin{proof}
At \(\theta=0\), take five scalar blocks
\begin{align}
  q_a(\theta)&=\frac15(1+s_a\theta),\notag\\
  (s_1,s_2,s_3,s_4,s_5)
  &=\frac15(-11,-6,-1,4,14).
  \label{eq:lloyd-counterexample-model}
\end{align}
Because \(\sum_as_a=0\), these probabilities sum to one; they are positive for sufficiently small \(|\theta|\).  Their scalar SLD scores at zero are \(s_a\).  For \(M=2\), consider
\begin{equation}
  C_1=\{1,2,3,4\},
  \qquad
  C_2=\{5\}.
\end{equation}
The centers are
\begin{equation}
  T_1=-\frac{7}{10},
  \qquad
  T_2=\frac{14}{5}.
\end{equation}
The midpoint between them is \(21/20\), while
\(s_1,s_2,s_3,s_4<21/20<s_5\); hence every branch is assigned to a nearest current center and the Lloyd iteration stops.  Its reduced cost is
\begin{equation}
  \overline{\mathcal D}
  =\frac15\left(\frac94+\frac14+\frac14+\frac94\right)=1.
\end{equation}
Move branch 4 to the second class and recompute the centers.  The new partition
\(C_1'=\{1,2,3\}\), \(C_2'=\{4,5\}\) has centers
\(-6/5\) and \(9/5\), and reduced cost
\begin{equation}
  \overline{\mathcal D}'
  =\frac15(1+0+1+1+1)=\frac45<1.
\end{equation}
Thus the Lloyd fixed point is not a Hamming-one local minimum of the reduced objective.
\end{proof}

\subsection{Exact one-swap refinement}

Suppose branch \(a\) currently belongs to class \(i\), and consider moving it to \(j\neq i\).  Let \(f^{a:i\to j}\) denote the resulting partition.  Equation~\eqref{eq:reduced-objective} shows that the exact change depends only on the two affected blocks:
\begin{align}
  \Delta_{a:i\to j}
  &\equiv\overline{\mathcal D}(f^{a:i\to j})
  -\overline{\mathcal D}(f)\notag\\
  &=\mathfrak F(\tau_i,\dot\tau_i)
  +\mathfrak F(\tau_j,\dot\tau_j)\notag\\
  &\quad-\mathfrak F(\tau_i-\tau_a,\dot\tau_i-\dot\tau_a)
  -\mathfrak F(\tau_j+\tau_a,\dot\tau_j+\dot\tau_a).
  \label{eq:exact-one-swap-change}
\end{align}
Starting from a Lloyd fixed point, evaluate Eq.~\eqref{eq:exact-one-swap-change} for every branch and every alternative class.  If a negative value exists, accept a move with the most negative value, recompute the two affected block SLDs, and repeat.

\begin{theorem}[Finite one-swap refinement]
\label{thm:one-swap-refinement}
The exact one-swap refinement terminates after finitely many accepted moves.  Its final partition is a Hamming-one local minimum of \(\overline{\mathcal D}\): no reassignment of a single branch to another class, followed by exact recentering, decreases the objective.
\end{theorem}

\begin{proof}
Every accepted move has \(\Delta_{a:i\to j}<0\) and therefore strictly decreases \(\overline{\mathcal D}\).  Since only finitely many labeled partitions exist, no partition can repeat and the procedure must terminate.  At termination, Eq.~\eqref{eq:exact-one-swap-change} is nonnegative for every Hamming-one neighbor, which is precisely the claimed local-optimality condition.
\end{proof}

\subsection{Complexity and initialization}

For Hilbert-space dimension \(d\) and \(r\) fine branches, precomputing the fine SLD scores costs \(O(rd^3)\) using dense diagonalization.  A Lloyd center update costs \(O(Md^3)\), while a direct assignment sweep costs at most \(O(rMd^3)\) and can be reduced by caching matrix products.  A complete one-swap sweep evaluates \(r(M-1)\) candidates and requires at most two new block-SLD solves per candidate, again \(O(rMd^3)\) without incremental caching.  These are arithmetic scaling estimates, not hardware runtime claims.

Because global optimization remains combinatorial, multiple initial partitions may be used.  A branch-based analogue of \(k\)-means++ selects the first seed uniformly and each subsequent seed with probability proportional to its current distance from the nearest selected seed.  This is an initialization heuristic only and carries no global-optimality guarantee here.

\end{document}